\documentclass[12pt]{article}
\usepackage[utf8]{inputenc}
\usepackage[dvipsnames]{xcolor}
\usepackage[numbers]{natbib}
\usepackage{amsthm,amsmath,amssymb,bbm,hyperref,graphicx,float,tikz,comment,mathrsfs,dsfont}
\usepackage[normalem]{ulem} 
\usepackage[affil-it]{authblk}
\allowdisplaybreaks

\newtheorem{prop}{Proposition}
\newtheorem{lemma}{Lemma}

\theoremstyle{definition}
\newtheorem{rmk}{Remark}

\newtheorem{ex}{Example}
\theoremstyle{remark}

\newcommand{\formost}{{\textstyle\bigvee \hspace{-3.04mm} | \,\,}}
\newcommand{\law}{\mathscr{L}}
\DeclareMathOperator{\tr}{tr}

\newcommand{\mc}{\mathrm{mc}}
\newcommand{\eq}{\mathrm{eq}}

\newcommand{\Hilbert}{\mathscr{H}}
\newcommand{\E}{\mathcal{E}}
\newcommand{\be}{\begin{equation}}
\newcommand{\ee}{\end{equation}}
\newcommand{\EEE}{\mathbb{E}}
\newcommand{\PPP}{\mathbb{P}}
\newcommand{\RRR}{\mathbb{R}}
\newcommand{\SSS}{\mathbb{S}}
\newcommand{\gc}{\mathrm{gc}}
\newcommand{\gmc}{\mathrm{gmc}}
\newcommand{\gG}{\mathrm{gG}}

\newcommand{\can}{\mathrm{can}}
\newcommand{\SSSH}{\mathbb{S}(\mathscr{H})}

\newcommand{\GAP}{\mathrm{GAP}}
\newcommand{\Fock}{\mathscr{F}}
\newcommand{\Sc}{\overline{S}}

\title{Grand-Canonical Typicality}
\author[1]{Cedric Igelspacher\thanks{ORCID: 0009-0006-0691-7030, E-mail: cedric.igelspacher@uni-tuebingen.de}}
\author[1]{Roderich Tumulka\thanks{ORCID: 0000-0001-5075-9929, E-mail: roderich.tumulka@uni-tuebingen.de}}
\author[2]{Cornelia Vogel\thanks{ORCID: 0000-0002-3905-4730, E-mail: cornelia.vogel@math.lmu.de}}
\affil[1]{Mathematics Institute, Eberhard Karls University T\"ubingen, Auf der Morgenstelle 10, 72076 T\"ubingen, Germany}
\affil[2]{Mathematics Institute, Ludwig Maximilians University,
Theresienstr. 39, 80333 Munich, Germany}
\date{May 21, 2026}

\begin{document}

\maketitle

\begin{abstract}
We study how the grand-canonical density matrix arises in macroscopic quantum systems. ``Canonical typicality'' is the known statement that for a typical wave function $\Psi$ from a micro-canonical energy shell of a quantum system $S$ weakly coupled to a large but finite quantum system $B$, the reduced density matrix $\hat{\rho}^S_\Psi=\tr^B |\Psi\rangle\langle \Psi|$ is approximately equal to the canonical density matrix $\hat{\rho}_\can=Z^{-1}_\can \exp(-\beta \hat{H}^S)$. Here, we discuss the analogous statement and related questions for the \emph{grand-canonical} density matrix $\hat{\rho}_\gc=Z^{-1}_\gc \exp(-\beta(\hat{H}^S-\mu_1 \hat{N}_{1}^S-\ldots-\mu_r\hat{N}_{r}^S))$ with $\hat{N}_{i}^S$ the number operator for molecules of type $i$ in the system $S$. This includes (i) the case of chemical reactions (which requires some novel considerations) and (ii) that of systems $S$ defined by a spatial region which particles may enter or leave. It includes statements about how $\hat{\rho}_\gc$ arises from the density matrix of the appropriate (generalized micro-canonical) Hilbert subspace $\Hilbert_\gmc \subset \Hilbert^S \otimes \Hilbert^B$ (defined by a micro-canonical interval of total energy and suitable particle number sectors) or from typical $\Psi$ in $\Hilbert_\gmc$, as well as statements about the distribution of the (conditional) wave function $\psi^S$ of $S$, which turns out to be a so-called GAP or Scrooge measure. That is, we discuss the foundation and justification of both the density matrix and the distribution of the wave function in the grand-canonical case. To this end (particularly for the chemical reactions), we also need to extend these considerations to the so-called generalized Gibbs ensembles, which apply to systems for which some macroscopic observables are conserved.

\medskip

\end{abstract}

\section{Introduction}
\label{sec:1}

We are concerned with quantum analogs of the micro-canonical, the canonical, and the grand-canonical ensemble introduced by Gibbs \cite{Gibbs1902} to classical statistical mechanics. While the micro-canonical and the canonical ones have been discussed extensively in the literature, we cover here the foundations of the grand-canonical one  (as well as the generalized Gibbs ensemble of \eqref{rhogGdef} below). Each of the three ensembles has two kinds of analogs in quantum theory: a density matrix and a probability distribution of wave functions (i.e., over the unit sphere in Hilbert space). The density matrices are well known:
\begin{subequations}
\begin{align}
\hat{\rho}_\mc &= Z_\mc^{-1} \; \mathds{1}_{[E-\Delta E,E]}(\hat{H}) \label{rhomcdef}\\[1mm]
\hat{\rho}_\can &= Z_\can^{-1} \; \exp(-\beta \hat{H}) \label{rhocandef}\\
\hat{\rho}_\gc &= Z_\gc^{-1} \; \exp\Bigl(-\beta(\hat{H}- \mu_1 \hat{N}_{1} - \ldots - \mu_r \hat{N}_{r})\Bigr) \label{rhogcdef}
\end{align}
\end{subequations}
are the micro-canonical, canonical, and grand-canonical density matrices of a quantum system with Hamiltonian $\hat{H}$ and, in the last case, particle number operator $\hat{N}_{i}$ for particle type $i=1,\ldots, r$; $Z$ is always the normalizing constant, $\beta$ the inverse temperature, $\mu_i$ the parameter called the chemical potential, and $\mathds{1}_{[E-\Delta E,E]}(\hat{H})$ is the projection to the energy shell or micro-canonical subspace $\Hilbert_{[E-\Delta E,E]}\subset \Hilbert$ (i.e., the spectral subspace of $\hat{H}$ for this interval, or the span of all eigenvectors with eigenvalue in this interval) in the Hilbert space $\Hilbert$. It is assumed in \eqref{rhogcdef} that
\be
\Bigl[\hat{H},\sum_i \mu_i \hat{N}_{i} \Bigr]=0,
\ee
or else $\hat{\rho}_{\gc}$ would not be invariant under the unitary time evolution $\exp(-i\hat{H}t/\hbar)$ and thus could not  represent thermal equilibrium (which would at least have to be stationary). 

The relevant distribution of the wave function in the unit sphere
\be
\SSS(\Hilbert) = \{\psi \in\Hilbert: \|\psi\|=1\}
\ee
of $\Hilbert$ is \cite{Bloch,GLTZ_06}
\begin{subequations}
\begin{align}\label{GAPmccan}
\PPP_\mc(d\psi) &= u_{\SSS(\Hilbert_{[E-\Delta E,E]})}(d\psi)\\
\PPP_\can(d\psi) &= \GAP_{\hat{\rho}_\can} (d\psi) \label{PcanGAP}
\end{align}
\end{subequations}
in the micro-canonical and canonical case, where $u_{\SSS}$ means the normalized uniform distribution over the sphere $\SSS$ (i.e., the surface area measure such that $u_{\SSS}(\SSS)=1$) and $\GAP_{\hat{\rho}}$ the \textit{Gaussian adjusted projected} (GAP) measure \cite{JRW_94,GLTZ_06,GLMTZ_15,TTV_24,Mark_24,Vog_24,MS_25,Pre26} with density matrix $\hat{\rho}$ (i.e., the most spread-out measure with density matrix $\hat\rho$, also known as the Scrooge measure, see Section~\ref{sec:GAP} for the definition). By \eqref{PcanGAP}, $\GAP_{\hat{\rho}_\can}$ describes the thermal equilibrium distribution of wave functions and can be viewed as a quantum analog of the canonical ensemble from classical statistical mechanics. For example, photons coming from the sun can be taken to have a GAP-distributed wave function. Among other things, a novelty of our paper is that we determine $\PPP_\gc$ (see Statements \hyperref[stat3]{3} and \hyperref[stat4b]{4b} in Section~\ref{sec:claims}). 

Here is a brief overview of the other contributions of this paper. While some of our statements are more or less straightforward extensions of well-known results \cite{PSW_06,GLTZ_06a}, most require novel considerations. In the literature, $\hat\rho_\can$ has been derived (i) from a maximum entropy principle, (ii) as the partial trace of $\hat\rho_\mc$, and (iii) by canonical typicality. As we explain, variant (i) is ultimately not a satisfactory justification. The extension of (ii) and (iii) to $\hat\rho_\gc$ hits obstacles in one of the two main cases in which particle numbers change: chemical equilibrium. It requires a generalization of $\hat\rho_\gc$ that has been called the \emph{generalized Gibbs ensemble} $\hat\rho_\gG$. After clarifying how $\hat\rho_\gG$ applies to chemical equilibrium, and after formulating as a ``general Gibbs principle'' in what sense these density matrices are valid, we formulate and derive general versions of several statements about how, and under which conditions, $\hat\rho_\gG$ and $\hat\rho_\gc$ arise as either the partial trace of a generalized micro-canonical density matrix or the reduced density matrix of a typical $\Psi$ from a generalized micro-canonical subspace. We thereby also obtain a unified treatment of changes in particle number due to either chemical reactions or particles entering and leaving a spatial region. We now explain these issues in the remainder of Section~\ref{sec:1}.

\subsection{How the Ensembles Arise}

The justification of the canonical and grand-canonical density matrices and distributions of wave functions has to do with weakly coupling the quantum system, let us call it $S$, to another quantum system $B$ that is much bigger than $S$ (but finite) and serves as a bath (reservoir). There is then a precise sense of what we mean by ``the'' probability distribution on $\SSS(\Hilbert^S)$ (where we write the system as an upper index): this sense is the typical distribution of the \emph{conditional wave function} \cite{LD90,DGZ_92,GN_99,GLTZ_06,TTV_24} (a concept first introduced in Bohmian mechanics \cite{LD90,DGZ_92} but applicable in any version of quantum mechanics); see Section~\ref{sec:3.2} for the definition and discussion. Put briefly, the conditional wave function of a system $S$ entangled with another system $B$ is the natural way of assigning a wave function $\psi^S$ to $S$ when the wave function $\Psi$ of the combined system is given: take the configuration $Q_B$ of $B$ to be random with the Born distribution and conditionalize on the value of $Q_B$ by setting $\psi^S(q_S) = \Psi(q_S,Q_B)$ times a normalizing factor. (More generally, a conditional wave function can also be defined relative to other bases than the position basis.)\footnote{\label{fn:cond}Conditional wave functions also have a deeper physical relevance as follows: if the time evolution of the full wave function $\Psi$ is taken to be unitary (as it is here), then macroscopic superpositions will arise; such situations are known particularly from quantum measurements, where one would assume that in the real world, a definite outcome exists, and the relevant wave function is the one conditioned on that outcome (say, conditioned on the pointer configuration of the apparatus---the conditional wave function). It has been argued \cite{DGZ_92} that whenever we physicists talk about a wave function that is not the wave function of the universe, we really mean the conditional wave function of some system.}

\bigskip

Coming back to density matrices, there are two well-known ways of justifying $\hat{\rho}_\can$: either as the partial trace over $B$ of the micro-canonical density matrix of $S\cup B$, or as the partial trace over $B$ of $|\Psi\rangle \langle \Psi|$, where the wave function $\Psi$ of $S\cup B$ is a typical unit vector in a micro-canonical subspace of $\Hilbert^S \otimes \Hilbert^B$ (a consideration known as \emph{canonical typicality} \cite{GM_03,PSW_06,GLTZ_06a}; note that the word ``typical'' refers to a probability distribution over $\mathbb{S}(\Hilbert^S\otimes\Hilbert^B)$, here either the uniform distribution $u$ on the unit sphere of a micro-canonical subspace or $\GAP_{\hat{\rho}}$.). Since the former density matrix is the average of the latter, the latter argument is the stronger one: it asserts that for \emph{most} $\Psi$, $S$ will appear to have density matrix $\hat{\rho}_\can$, whereas the former argument yields $\hat{\rho}_\can$ for $S$ only \emph{on average}. These two ways arise as well in the grand-canonical case. However, some considerations specific to this case are required. Specifically, Statements~\hyperref[stat2a]{2a}, \hyperref[stat2b]{2b} and \hyperref[stat3]{3} in Section~\ref{sec:claims} follow from theorems of Popescu et al.~\cite{PSW_06} and Goldstein et al.~\cite{GLMTZ_15}, while for Statements~\hyperref[stat1a]{1a}, \hyperref[stat1b]{1b}, \hyperref[stat4a]{4a}, \hyperref[stat4b]{4b}, and \hyperref[stat5]{5}, new approaches are needed.

In particular, our considerations are in line with the ``individualist'' attitude \cite{Gold_19,GLTZ_20} that a closed system $S\cup B$ in a pure state $\Psi$ can display thermodynamic behavior; this attitude has been widely applied in recent years, in particular in connection with the eigenstate thermalization hypothesis (ETH) (e.g., \cite{Sre94, LPSW, GLMTZ_10, GE_16, Mori_18}). Also the distributions $\PPP_\mc,\PPP_\can,\PPP_\gc$ arise here from an \emph{individual} pure state $\Psi$ (in a way to which the conditional wave function is relevant). 

A standard way of arriving at either $\hat\rho_\can$ or $\hat\rho_\gc$ is to maximize the von Neumann entropy 
\be\label{SvN}
S_\mathrm{vN}=-k_\mathrm{B}\tr(\hat\rho \log \hat\rho)
\ee
(with Boltzmann's constant $k_{\mathrm{B}}$) 
under the constraint that the expected energy is fixed, $\tr(\hat\rho \, \hat{H})=E$ (as well as, in the grand-canonical case, the expected particle numbers, $\tr(\hat\rho\, \hat{N}_i)=n_i$). While this reasoning may suggest that $\hat\rho_\can$ ($\hat\rho_\gc$) represents the knowledge of an observer who knows not more about the state than $E$ (and $n_i$) (e.g., \cite{Jay_57,Wu97}, \cite[Chap.~4]{Bal06}, \cite[Chap.~7]{Bri22})
the reasoning outlined above in terms of $S\cup B$ provides more \cite{PSW_06}: a physical reason why the actual, physical state (of $S$) is given (approximately) by $\hat\rho_\can$ (or $\hat\rho_\gc$); see Section~\ref{sec:comparison} for further discussion.

Also concerning the use of the micro-canonical distribution $u_\mc=u_{\SSS(\Hilbert_\mc)}$ or density matrix $\hat\rho_\mc$, there are different attitudes (e.g., \cite[Sec.~6.8]{Bri22}) in the literature: Sometimes its use is regarded as a ``principle of equal a priori probabilities'' (e.g., \cite{Tol_38,Jan_69,Pen_70}) saying that unless we have further knowledge about $\Psi$ than that it lies in $\Hilbert_\mc$, it is a rational guess to assign each point on $\SSS(\Hilbert_\mc)$ equal (subjective) probability. As we elucidate in Section~\ref{sec:root}, our attitude is different and can be summed up like this: Since it is the nature of thermal equilibrium to behave like most $\Psi$, we find the thermal equilibrium behavior by studying the behavior of most $\Psi$, where ``most'' refers to $u_\mc$ (or, when appropriate, $u_\gmc$ over $\SSS(\Hilbert_\gmc)$ as in \eqref{Hilbertgmcdef}).

\subsection{Chemical Equilibrium and General Gibbs Principle}

Another aspect in the discussion of the grand-canonical ensemble arises from the fact that it refers to a variable particle number and, already classically, there are two distinct situations in which that comes up:  (i)~for chemical reactions such as
\be\label{reaction1}
A + B \leftrightharpoons C + \delta E
\ee
(``chemical equilibrium''), and (ii)~when particles can enter and leave a system (defined by a region in space). To illustrate the latter case, thermal equilibrium in the absence of external fields then leads to equal particle densities in the region and its complement; in order to have a name for it, we will call it ``spatial equilibrium'' in contrast to ``chemical equilibrium.''   
Many textbooks (e.g., \cite{Tol_38,LL_80, Bloch, Schw_06}) focus on spatial equilibrium when discussing the grand-canonical ensemble and omit proper reasoning for chemical equilibrium, or focus on thermodynamic potentials when discussing chemical reactions and omit proper justification from the microscopic physical laws. Here, we provide this connection also for chemical equilibrium, which, to the best of our knowledge, is novel. We suggest that chemical equilibrium is best understood by considering conserved macroscopic observables (such as $\hat{N}_A-\hat{N}_B$ and $\hat{N}_C+\hat{N}_A$ for the chemical reaction mentioned above) and the more abstract perspective of the so-called ``generalized Gibbs ensemble'' \cite[(4.6)]{Bal06}, \cite{RDYO_07, Rigol2, Rigol3, Mori_18}
\be\label{rhogGdef}
\hat{\rho}_\gG = Z^{-1}_\gG \; \exp\left(\sum_{k=1}^K \lambda_k \hat{Q}_k\right) \,,
\ee
which we also discuss as it plays, in fact, a key role. It represents the thermal equilibrium ensemble for which $\hat{Q}_1,\ldots, \hat{Q}_K$ (and only those) are conserved macroscopic observables,\footnote{The expression ``macroscopic observable'' is often intended to include the assumption that the eigenvalues of $\hat{Q}_k$ are separated by the resolution of macroscopic measurements; this is not intended here.} 
assuming they commute with each other. (One of them usually is the Hamiltonian $\hat{H}$. Preliminary considerations in this direction can be found in \cite[Sec.~I.4]{LL_80}.) Note that, as a consequence, there is a $K$-parameter family of thermal equilibrium states with parameters $\lambda_1,\ldots,\lambda_K$ (which includes the canonical case with $K=1$, $\hat{Q}_1=\hat{H}$, and $\lambda_1=-\beta$). The \emph{generalized micro-canonical} subspace $\Hilbert_\gmc$ for this situation is the one where the eigenvalues of each $\hat{Q}_k$ are restricted to an interval $[Q_k-\Delta Q_k, Q_k]$ that is short on the macroscopic scale but still has a high-dimensional spectral subspace $\mathds{1}_{[Q_k-\Delta Q_k,Q_k]}(\hat{Q}_k)\Hilbert$,
\begin{subequations}\label{Hilbertgmcdef}
\begin{align}
\Hilbert_\gmc &:= \mathrm{span}\Bigl\{\phi\in \Hilbert~:~ \forall k: \hat{Q}_k\phi=q_k\phi,~ q_k\in [Q_k-\Delta Q_k,Q_k]  \Bigr\}\\
&\:= \prod_{k=1}^K \mathds{1}_{[Q_k-\Delta Q_k,Q_k]}(\hat{Q}_k) \Hilbert\,.
\end{align}
\end{subequations}
We write
\be\label{rhogmcdef}
\hat\rho_\gmc=\frac{\hat{P}_\gmc}{\tr \hat{P}_\gmc}
\ee
for the normalized projection to $\Hilbert_\gmc$. The key property, which we call the ``General Gibbs Principle'', can now be formulated as follows: 

\bigskip

\noindent{\bf General Gibbs Principle.} {\it Suppose the self-adjoint operators $\hat{Q}_1,\ldots,\hat{Q}_K$ commute with each other, and $\dim\Hilbert_\gmc$ is large. Then, in relevant senses of equivalence of ensembles, $\hat\rho_\gmc$ is equivalent to $\hat\rho_\gG$ with $\lambda_k$ chosen so that
\be\label{fixlambda}
\tr(\hat\rho_\gG \, \hat{Q}_k) = Q_k\,.
\ee
}

\bigskip

(For $K=1$ and $\hat{Q}_1=\hat{H}$, we have that $\hat{\rho}_\gG=\hat\rho_\can$ and $\hat\rho_\gmc=\hat\rho_\mc$, and the general Gibbs principle asserts that the canonical density matrix is equivalent to the micro-canonical one, provided we define $\beta$ through the relation $\tr(\hat{H} e^{-\beta \hat{H}})/\tr(e^{-\beta\hat{H}})=E$ with given energy $E$.)

\bigskip

One relevant sense of equivalence of ensembles is that, as conjectured by \cite{Rigol2}, they yield the same Born distribution for practically all observables $\hat{A}$ except very specially chosen functions of the $\hat{Q}_k$. We elucidate in Section~\ref{sec:plausibility} why this is plausible.

We focus on another sense of equivalence: We say that for a composite system $S\cup B$, two density matrices 
\be
\text{$\hat\rho_1$ and $\hat\rho_2$ are \emph{$S$-equivalent}}~\Leftrightarrow~
\tr^B \hat\rho_1 \approx \tr^B \hat\rho_2\,,
\ee
where $\tr^B$ means the partial trace over $\Hilbert^B$.\footnote{We can take $\approx$ to mean (for the sake of definiteness) closeness in the trace norm. However, other senses of closeness could be considered as well, and most of our considerations are not on the level of mathematical rigor that would correspond to fixing a particular norm.} We define further that
\be
\begin{aligned}
&\text{$\hat\rho_1$ and $\hat\rho_2$ are \emph{locally equivalent}} ~\Leftrightarrow~\\
&\text{they are $S$-equivalent for every  macro-small spatial region $S$,}
\end{aligned}
\ee
where we say of a spatial region $S\subset \RRR^3$ that
\be\label{macrosmalldef}
\begin{aligned}
&\text{$S$ is \emph{macro-small}}~\Leftrightarrow\\ 
&\text{its volume is small (on the macroscopic scale),}\\ 
&\text{$S$ is large compared to the size of a molecule, and}\\
&\text{$S$ has moderate surface.}
\end{aligned}
\ee
Here, we say of a spatial region $S$ that
\be\label{moderatedef}
\begin{aligned}
    &\text{$S$ has \emph{moderate surface}}~\Leftrightarrow\\
    &\text{surface area$(S)\leq C\,\mathrm{vol}(S)^{2/3}$ with (say) $C=100$.}
\end{aligned}
\ee
This condition expresses that the shape of $S$ is not too frayed or ``spaghetti-like''; examples of shapes with moderate surfaces include balls ($C\approx 4.8$) and cubes ($C=6$), with non-moderate surfaces cylinders that are very long (``cigars'') or very oblate (``pancakes'').

\bigskip

\noindent{\bf Conjecture.} Except for very special choices of $\hat{Q}_k$,
\be\label{conjecture}
\hat\rho_\gG \text{ is locally equivalent to } \hat\rho_\gmc\,.
\ee

\bigskip

We give in this paper a derivation of \eqref{conjecture} in the case in which each of the $\hat{Q}_k$ is (approximately) \emph{extensive}, that is, in which, for any spatial region $S$ with complement $\Sc$,
\be\label{Qkextensive}
\hat{Q}_k \approx \hat{Q}_{k}^S \otimes \hat{I}^{\Sc} + \hat{I}^S \otimes \hat{Q}_k^{\Sc} \,,
\ee
provided $S$ is large compared to the size of a molecule and has moderate surface.\footnote{If we want to make $\approx$ more precise, we may think of closeness in the operator norm; but this is not very relevant for the kind of outline provided here.} This is the first of our claims, see Statement~\hyperref[stat1a]{1a} in Section~\ref{sec:claims}. It includes the grand-canonical density matrix $\hat\rho_\gc$ as a special case (also for chemical equilibrium, see Section~\ref{sec:chemical}, but also for spatial equilibrium, where the conserved observables are simply $\hat{Q}_k=\hat{N}_k$).

There are further question about equivalence of ensembles that deserve a deeper discussion, such as how large or small the width $\Delta E$ of the energy window $[E-\Delta E,E]$ can be chosen, or about the size of fluctuations, which can matter near phase transitions; however, we will not provide such a discussion here.

\subsection{Approach to Equilibrium}

We will discuss the \emph{approach} (time evolution) toward the grand-canonical and the generalized Gibbs ensemble, both concerning the reduced density matrix and the distribution of the conditional wave function; we show that this approach occurs for every initial wave function if the Hamiltonian $\hat{H}$ satisfies the appropriate version of the eigenstate thermalization hypothesis (ETH), see Eq.s \eqref{ETHa} and \eqref{ETHb} in Statement~\hyperref[stat5]{5}. For the reduced density matrix, we also prove a mathematical theorem (Proposition~\ref{prop: MITE} in Section~\ref{sec:approach}) that somewhat generalizes the existing results in this direction
\cite{Sre_95,Sre_96,RDO_08,GHLT_17,Mori_18} (which were intended for the canonical density matrix). Concerning the approach to thermal equilibrium of the conditional wave function, no results have been in the literature before for either the canonical or the grand-canonical case.

\bigskip

The remainder of this paper is organized as follows. In Section~\ref{sec:chemical}, we elucidate the framework that will also cover chemical reactions. In Section~\ref{sec:claims}, we collect the main statements that we derive in this paper. In Section~\ref{sec:discussion}, we compare ours to other derivations. In Sections~\ref{sec: deriv 1a}--\ref{sec:approach}, we justify and discuss these statements. Appendices \ref{app:differentF}--\ref{app:ergodic} collect the mathematical proofs.

\section{Use of $\hat\rho_\gc$: Chemical Reactions}
\label{sec:chemical}

While our statements in Section~\ref{sec:claims} also apply to spatial equilibrium, their full meaning unfolds itself in the case of chemical equilibrium. Since the latter is not commonly discussed in the way we treat it here, we first set the stage in this section by describing our perspective on the use of $\hat\rho_\gc$ for chemical reactions. Specifically, we explain how the general Gibbs principle formulated around \eqref{fixlambda} can be applied here, yields $\hat\rho_\gc$, and allows us to determine the chemical equilibrium. On that basis, we can then formulate our statements in Section~\ref{sec:claims}.

Consider a macroscopic quantum system containing $r$ different chemical substances $A_1,\ldots, A_r$, and let $n_i$ with $i=1,\ldots,r$ stand for the number of molecules of substance $A_i$. As a general version of \eqref{reaction1}, suppose there are $L\in\mathbb{N}$ possible chemical reactions, given by 
\be\label{reaction2}
\nu_{\ell 1} A_1+\ldots + \nu_{\ell r} A_r \leftrightharpoons
\tilde{\nu}_{\ell 1} A_1+\ldots + \tilde{\nu}_{\ell r} A_r +\delta E_\ell
\ee
with $\ell=1,\ldots,L$, $\nu_{\ell i},\tilde{\nu}_{\ell i} \in \{0,1,2,\ldots\}$ the number of molecules $A_i$ involved, and $\delta E_{\ell}$ the amount of energy (positive, negative, or possibly zero) released (or, if negative, consumed) in this reaction. Suppose initially the system contains $n_{0i}$ molecules $A_i$ for each $i=1,\ldots,r$ and has energy in the micro-canonical interval $[E-\Delta E, E]$. The questions arise:
\be\label{question1}
\text{What are the numbers $n_{\eq,i}$ of $A_i$ in chemical equilibrium?}
\ee
And
\be\label{question2}
\text{How can grand-canonical density matrices be used for calculating them?}
\ee
We answer both in this section. Our reasoning applies not only to chemical reactions but also to ionization and to elementary particle reactions.

The appropriate Hilbert space is
\be
\Hilbert=\Fock_1\otimes \cdots \otimes \Fock_r
\ee
with
\be
\Fock_i = \bigoplus_{n_i=0}^\infty \Fock_i^{(n_i)}
\ee
the Fock space of the molecules of type $i$ and $\Fock_i^{(n_i)}$ the sector with $n_i$ molecules; it is a fermionic (bosonic) Fock space whenever the number of fermions in the molecule is odd (even), which occurs for an electrically neutral molecule (i.e., one whose number of electrons equals the number of protons) whenever the number of neutrons per molecule is odd (even). (Note that the 1-particle Hilbert space $\Hilbert_{1i}$ from which $\Fock_i$ arises may also involve internal degrees of freedom.)

\begin{rmk}
There may be a (somewhat arbitrary) choice of the theoretical physicist involved in defining which states exactly count (say, in the example of \eqref{reaction1}) as states of molecule $C$ as opposed to states of a molecule $A$ and a molecule $B$ \cite[Sec.~3]{Leb22}. We assume in the following that such a choice has been made. In the case of ionization (say, $A$ = proton, $B$ = electron, $C$ = hydrogen atom), a natural choice \cite{Gir90} would be to count the bound states (from the discrete spectrum of the relative Hamiltonian) as atoms $C$ and the scattering states (from the continuous spectrum of the relative Hamiltonian) as $A+B$.\hfill$\diamond$
\end{rmk}

Of the Hamiltonian $\hat{H}$, we assume that molecules, when separated by a sufficient (still microscopic) distance, hardly interact. Since the details of the interaction term do not matter much in the end, let us focus for a moment on the part of $\hat{H}$ that does matter, the part we need to write down for calculating $n_{\eq,i}$. A key contribution is the kinetic energy of the center of mass of each molecule; another one the external field (e.g., electric, magnetic, or gravitational) to which each molecule is subject; another one the energy contributions of internal degrees of freedom (e.g., rotational, vibrational, or librational) of each molecule. We write $\hat{H}_{1i}$ for the 1-molecule Hamiltonian combining these contributions, $\hat{H}_{0i}$ for the second quantized version of $\hat{H}_{1i}$ (i.e., the operator on $\Fock_i$ acting like $\hat{H}_{1i}$ on each particle), and $\hat{H}_0$ for the combination of those over all types $i$,
\be\label{H0H0i}
\hat{H}_0 := \sum_{i=1}^r \hat{I}_1\otimes \cdots \otimes \hat{I}_{i-1} \otimes \hat{H}_{0i} \otimes \hat{I}_{i+1}\otimes \cdots \otimes \hat{I}_r
\ee
with $\hat{I}_i$ the identity on $\Fock_i$. Another key contribution to the full Hamiltonian $\hat{H}$ comes from the \emph{rest energy} of each molecule; that is, apart from the relativistic contribution from the rest mass of the particles (which is irrelevant for chemical reactions because electron-positron or quark-antiquark pair creation do not occur), the \emph{ground state energy} $E_{0i}$ of each molecule. For example in the reaction \eqref{reaction1}, $\delta E$ represents the binding energy between $A$ and $B$ within a $C$ molecule, meaning that 
\be\label{E0C}
E_{0C}=E_{0A}+E_{0B}-\delta E\,.
\ee
Likewise in the general case \eqref{reaction2}, the $E_{0i}$ have to satisfy
\be\label{nuE}
\nu_{\ell 1}E_{01} + \ldots + \nu_{\ell r}E_{0r} = \tilde\nu_{\ell 1}E_{01} + \ldots + \tilde\nu_{\ell r}E_{0r} +\delta E_\ell
\ee
for every $\ell=1,\ldots,L$. Up to some remaining freedom, the $E_{0i}$ can be determined from the known $\delta E_\ell,\nu_{\ell i},\tilde\nu_{\ell i}$ by solving \eqref{nuE}; if a solution to \eqref{nuE} did not exist, it would signal the existence of a circle of reactions from \eqref{reaction2} that produce energy while returning the same number of each molecule, which is physically impossible. We will see in Remark~\ref{rmk:freedom} below that the freedom in the $E_{0i}$ does not affect either $\hat\rho_\gc$ or $n_{\eq,i}$. Suppose we know the $E_{0i}$; we can then define
\be\label{H*def}
\hat{H}_* := \hat{H}_0+\sum_{i=1}^r E_{0i} \, \hat{N}_i\,,
\ee
which is equivalent to adding $E_{0i}$ to the 1-particle Hamiltonian $\hat{H}_{1i}$ before the second quantization. 

Of the full Hamiltonian again,
\be
\hat{H}= \hat{H}_* + \hat{V}\,,
\ee
we assume that the interaction term $\hat{V}$ is mostly negligibly small\footnote{It would be of interest to study in more detail which limit or sense of smallness of $\hat{V}$ is needed; we use here two assumptions: that $\hat{V}$ can be neglected in a certain step of the calculation, from \eqref{31} to \eqref{32}, and that it cannot be neglected for the question of which operators $\hat{Q}_k$ commute with $\hat{H}$.}
(namely, except during a chemical reaction) and contains non-zero transition elements corresponding to every reaction \eqref{reaction2}, but not for any transition that is not of the form \eqref{reaction2}. That is, we assume that the conserved quantities of \eqref{reaction2} are the only macroscopic conserved observables. In detail, using (simple) tools from linear algebra, consider the $r$-dimensional real space $\mathcal{R}$ with axes labeled $n_1,\ldots,n_r$, and let $\mathcal{L}$ be the subspace spanned by the $L$ vectors $(\nu_{\ell 1}-\tilde{\nu}_{\ell 1}, \ldots, \nu_{\ell r}-\tilde\nu_{\ell r})$. Any change in particle numbers due to reactions of the form \eqref{reaction2} lies in $\mathcal{L}$, any conserved linear combination of $n_1,\ldots,n_r$ is orthogonal to $\mathcal{L}$. (In fact, the deeper reason we allow real rather than integer components in $\mathcal{R}$ is that we allow real coefficients in these linear combinations.) Choose a complete set $F_1,\ldots,F_{K-1}$ of conserved linear combinations with 
\be
K=\dim(\mathcal{L}^\perp)+1=r+1-\dim \mathcal{L}
\ee
(which is equal to $r+1-L$ if the $L$ vectors just mentioned are linearly independent). That is, choose
a linear mapping (a $(K-1)\times r$-matrix) $F:\mathcal{R} \to \mathbb{R}^{K-1}$ with kernel 
\be\label{Fkernel}
F^{-1}(0)=\mathcal{L} \,.
\ee
By the dimension formula
\be
\dim \mathrm{kernel}(F) + \dim \mathrm{image}(F) = r\,,
\ee
$F$ has full rank, $\mathrm{rank}(F)=\dim \mathrm{image}(F)=K-1$, and for every $k=1,\ldots,K-1$, $F_k(n_1,\ldots,n_r)$ is conserved during every reaction \eqref{reaction2}. (Put differently, the rows $F_1,\ldots,F_{K-1}$ of the matrix $F$ form a basis of $\mathcal{L}^\perp$.) As a consequence, the observables
\be\label{Qkdef}
\hat{Q}_k := F_k(\hat{N}_1,\ldots, \hat{N}_r)
\ee
are conserved, i.e., they commute with $\hat{H}$. They also commute with each other because $\hat{N}_1,\ldots,\hat{N}_r$ do. Any other conserved linear combinations of number operators are linear combinations of the $\hat{Q}_k$ because we assumed that the transitions allowed by $\hat{H}$ between different sectors $\Fock_1^{(n_1)}\otimes \cdots \otimes \Fock_r^{(n_r)}$ are exactly those allowed by the reactions \eqref{reaction2}. We assume that there are no further conserved macroscopic observables. 

Now set $\hat{Q}_K:= \hat{H}$; consider the generalized micro-canonical subspace $\Hilbert_\gmc$ as defined in \eqref{Hilbertgmcdef} with $Q_k=F_k(n_{01},\ldots,n_{0r})$ (and suitably small $\Delta Q_k$) for $k=1,\ldots,K-1$ and $Q_K=E$, $\Delta Q_K=\Delta E$; consider further the
generalized Gibbs density matrix $\hat\rho_\gG$ as defined in \eqref{rhogGdef} with the values of $\lambda_k$ chosen to satisfy \eqref{fixlambda}. 

\begin{rmk}\label{rmk:differentF}
If we had chosen $F$ differently, it would not have changed $\hat\rho_\gG$.
In contrast, if we had chosen $F$ differently, it would have changed $\hat\rho_\gmc$ (though presumably not in a very relevant way). We prove both statements in Appendix~\ref{app:differentF}.\hfill$\diamond$
\end{rmk}

Now we apply the general Gibbs principle of Section~\ref{sec:1}. We obtain that there is a $K$-parameter family of thermal equilibrium states. (For example in \eqref{reaction1}, there is a 3-parameter family of thermal equilibrium states; the three parameters can be taken to be the energy (or, for that matter, temperature), $n_A-n_B$, and $n_A+n_C$.) 
The thermal equilibrium states can be represented by the generalized Gibbs density matrix
\be\label{31}
\hat\rho_\gG= Z_\gG^{-1} \, \exp\left(\sum_{k=1}^{K-1}\lambda_k \hat{Q}_k + \lambda_K \hat{H}\right)
\ee
with real parameters $\lambda_1,\ldots,\lambda_K$.
We write $-\beta$ for $\lambda_K$; since $\hat{H}$ is bounded from below and unbounded, $\lambda_K$ must be negative and thus $\beta$ positive. Since $\hat{H}=\hat{H}_* +\hat{V}$ with $\hat{V}$ small, we can replace $\hat{H}$ by $\hat{H}_*$,
\be\label{32}
\hat\rho_\gG\approx Z_\gG^{-1} \, \exp\left(\sum_{k=1}^{K-1}\lambda_k \hat{Q}_k -\beta \hat{H}_*\right)\,.
\ee
Since $\hat{Q}_k=F_k(\hat{N}_1,\ldots,\hat{N}_r)$ and $F_k$ is a linear function 
\be
F_k(n_1,\ldots,n_r)=\sum_{i=1}^r F_{ki} n_i\,,
\ee
we can write $\hat{Q}_k=\sum_{i=1}^r F_{ki}\hat{N}_i$ and
\be\label{rhogGH*}
\hat\rho_\gG\approx Z_\gG^{-1} \, \exp\Bigl(-\beta \bigl(\hat{H}_*-\sum_{i=1}^r \mu_{*i}\hat{N}_i \bigr)\Bigr)
\ee
with
\be
\mu_{*i} =\beta^{-1}\sum_{k=1}^{K-1}\lambda_k  F_{ki}\,.
\ee
By \eqref{H*def},
\be\label{finallyrhogc}
\hat\rho_\gG\approx \hat{\rho}_\gc=Z_\gc^{-1} \, \exp\Bigl(-\beta \bigl(\hat{H}_0-\sum_{i=1}^r \mu_{0i}\hat{N}_i \bigr)\Bigr)
\ee
with $Z_\gc=Z_\gG$ and
\be\label{mu0mu*}
\mu_{0i} =\mu_{*i} - E_{0i}\,.
\ee
The right-hand side of \eqref{finallyrhogc} is exactly what is meant by the formula \eqref{rhogcdef} for the grand-canonical density matrix $\hat\rho_\gc$ that arises here as a special case of $\hat\rho_\gG$. The $r+1$ coefficients $\mu_{01},\ldots,\mu_{0r},\beta$ are determined by the $K$ conditions \eqref{fixlambda} together with the $L$ conditions 
\be\label{munucond}
\sum_{i=1}^r (\mu_{0i}+E_{0i}) (\nu_{\ell i}-\tilde{\nu}_{\ell i})=0 \,,
\ee
which follow from $F^{-1}(0)\supseteq\mathcal{L}$ (and of which $\dim \mathcal{L}$ many are independent). 
By \eqref{H0H0i}, \eqref{finallyrhogc} can be rewritten as
\be\label{rhotensor}
\hat\rho_\gc= \bigotimes_{i=1}^r \hat\rho_i := \bigotimes_{i=1}^r Z_i^{-1} \, \exp\bigl( -\beta (\hat{H}_{0i}-\mu_{0i}\hat{N}_i) \bigr)\,.
\ee
If we write $\hat{H}_{1i}$ for the 1-particle Hamiltonian whose second quantization is $\hat{H}_{0i}$, then each tensor factor in \eqref{rhotensor} is a canonical density matrix with $\hat{H}_{0i}$ replaced by the second quantization of $\hat{H}_{1i}-\mu_{0i} \hat{I}$. 
This fact also provides a direct interpretation of the physical meaning of $\mu_{0i}$: The thermal density matrix looks as if it were canonical and molecule $i$ had ground state energy $-\mu_{0i}$.

Finally, the desired numbers $n_{\eq,i}$ characterizing the chemical equilibrium are then given by
\begin{subequations}\label{neqirhogc}
\begin{align}
n_{\eq,i} 
&= \tr(\hat\rho_\gmc \hat{N}_i)\\
&\approx \tr(\hat\rho_\gc \hat{N}_i)\\ 
&= \tr_{\Fock_i}(\hat\rho_i \hat{N}_i) = \frac{\tr\bigl(\hat{N}_i \exp(-\beta (\hat{H}_{0i}-\mu_{0i}\hat{N}_i))\bigr)}{\tr\exp(-\beta (\hat{H}_{0i}-\mu_{0i}\hat{N}_i))} \label{37c}\\
&= \frac{1}{\beta} \frac{\partial}{\partial \mu_{0i}}\log Z_\gc(\beta,\mu_{01},\ldots,\mu_{0r})
\end{align}
\end{subequations}
using \eqref{H0H0i}; the function $-(1/\beta)\log Z_\gc(\beta,\mu_{01},\ldots,\mu_{0r})$ is known as the \emph{grand potential} (e.g., \cite[(5.73)]{Bal06}, \cite[(2.7.13)]{Schw_06}). We have thus provided a justified answer to the questions \eqref{question1} and \eqref{question2}. Note that, while \eqref{37c} is an expression that one might have guessed right from the start, our derivation has achieved two things: First, provided reasons for why \eqref{neqirhogc} is true, and second, provided a way of computing the $\mu_{0i}$ from the known quantities $n_{0i},\delta E_\ell$ by means of \eqref{munucond} and \eqref{fixlambda}.

\begin{rmk}\label{rmk:freedom}
    Different values of $E_{0i}$ solving \eqref{nuE} would lead to the same $\hat\rho_\gG\approx\hat\rho_\gc$ (neglecting $\hat{V}$) and the same $n_{\eq,i}$; 
    they also lead to the same $\hat\rho_\gmc$, provided $\Delta Q_k=0$ and $Q_k$ is an eigenvalue of $\hat{Q}_k$ for $k=1,\ldots,K-1$.\footnote{Otherwise, $\hat\rho_\gmc$ still stays \emph{approximately} the same, provided the density of states of $\hat{Q}_k$ is a quickly increasing function, so most eigenvectors have eigenvalue near the right end of $[Q_k-\Delta Q_k, Q_k]$.} We prove this in Appendix~\ref{app:freedom}.\hfill$\diamond$
\end{rmk}

In the example of \eqref{reaction1}, $E_{0A}$ and $E_{0B}$ can be chosen arbitrarily, $E_{0C}$ is determined by \eqref{E0C}, \eqref{munucond} reduces to $\mu_{0C}=\mu_{0A}+\mu_{0B}+\delta E$, and \eqref{fixlambda} reduces to the three conditions (with $\hat\rho_A$ etc.\ as in \eqref{rhotensor}) 
\begin{subequations}
\begin{align}
    \tr(\hat\rho_A \hat{N}_A) - \tr(\hat\rho_B \hat{N}_B) &=n_{0A}-n_{0B}\\
    \tr(\hat\rho_A \hat{N}_A) + \tr(\hat\rho_C \hat{N}_C) &=n_{0A}+n_{0C}\\
    \sum_{i=A,B,C}\Bigl[\tr(\hat\rho_i \hat{H}_{0i})+E_{0i} \tr(\hat\rho_i\hat{N}_i) \Bigr]&=E\,,
\end{align}
\end{subequations}
which together determine $\mu_{0A},\mu_{0B}$, and $\beta$. (Note that $E$ would change if we changed $E_{0A},E_{0B}$, see Remark~\ref{rmk:freedom}.)

\begin{rmk}
The reaction rates depend on $\hat{V}$, but the equilibrium state and the $n_{\eq,i}$ do not, as long as $\hat{V}$ is small. Related questions for non-small $\hat{V}$ have been studied in \cite{Fef85,CLY89}.\hfill$\diamond$
\end{rmk}

\section{Main Claims}
\label{sec:claims}

We now present our main statements. Apart from Statement~\hyperref[stat1b]{1b}, the derivations of the statements are given below in the following sections.\footnote{Note that while the derivations (and therefore also the statements) are not mathematically rigorous, we always call the ones that are rigorous a proof. Generally, while some of the arguments we use in the derivations are mathematically precise, and some could be made precise, others are on a more heuristic level.}

\subsection{Reduced Density Matrix}

Let $\Lambda \subset \RRR^3$ denote the available volume of the quantum system considered and
\be
\Sc:=\Lambda \setminus S
\ee
the spatial complement of a subset $S$ (\emph{not} the closure of a set). We take $\hat\rho_\gmc$ as defined in \eqref{rhogmcdef} and $\hat\rho_\gG$ as defined in \eqref{rhogGdef}. 
For extensive $\hat{Q}_k$ as in \eqref{Qkextensive}, it follows that
\be\label{rhogGSSc}
    \hat\rho_\gG \approx \hat\rho_\gG^S \otimes \hat\rho_\gG^{\Sc}
\ee
with $\hat\rho_\gG^S$ formed according to  \eqref{rhogGdef} from the $\hat{Q}_k^S$ and likewise for $\Sc$. This also entails that
\be\label{trScrhogG}
    \tr^{\Sc} \hat\rho_\gG \approx \hat\rho_\gG^S \,.
\ee

\begin{itemize}
    \item [] \textbf{Statement 1a. (Generalized Gibbs Density Matrix)}\label{stat1a} {\it Let $\Lambda$ be a subset of $\mathbb{R}^3$ with moderate surface as in \eqref{moderatedef}, $\Hilbert=\Hilbert^\Lambda$ a Hilbert space, and for every $S\subseteq \Lambda$ let $\Hilbert^S$ be a Hilbert space such that for $S_1\subseteq S_2\subseteq \Lambda$, $\Hilbert^{S_2}$ factorizes into $\Hilbert^{S_2}=\Hilbert^{S_1}\otimes \Hilbert^{S_2\setminus S_1}$. Suppose the self-adjoint operators $\hat{Q}_1,\ldots,\hat{Q}_K$ on $\Hilbert$ are bounded from below and commute with each other, and that each is (approximately) extensive as in \eqref{Qkextensive}.
    Then for $\lambda_1,\ldots,\lambda_K$ satisfying \eqref{fixlambda}, large $\dim \Hilbert_\gmc$, and macro-small $S\subset\Lambda$,
    \be\label{1a}
    \tr^{\Sc} \hat\rho_\gmc \approx \hat\rho_\gG^S \,.
    \ee
    In particular, by \eqref{trScrhogG}, \eqref{1a} entails that $\hat\rho_\gmc$ and $\hat\rho_\gG$ are $S$-equivalent, and thus that $\hat\rho_\gmc$ and $\hat\rho_\gG$ are locally equivalent.
    }
\end{itemize}

We give a derivation in Section~\ref{sec: deriv 1a}.

\begin{itemize}
    \item [] \textbf{Statement 1b. (Grand-Canonical Density Matrix)} \label{stat1b} {\it Let $\Hilbert=\Fock_1\otimes \cdots \otimes \Fock_r$ with $\Fock_i$ the bosonic or fermionic Fock space over the 1-particle space $\Hilbert_{1i}$ and $\Lambda\subset \mathbb{R}^3$ with moderate surface. Suppose
    \be\label{1b1}
    \hat{H}\approx \hat{H}_*=\hat{H}_0+\sum_{i=1}^r E_{0i} \hat{N}_i
    \ee
    with $\hat{H}_0$ as in \eqref{H0H0i} and $\hat{H}_{0i}$ the second quantization of a 1-particle Hamiltonian $\hat{H}_{1i}$ that is bounded from below. Suppose for every $S\subseteq \Lambda$, 
    \be\label{Hilbert1additive}
    \Hilbert_{1i}=\Hilbert_{1i}^S\oplus \Hilbert_{1i}^{\Sc}
    \ee
    and every $S\subseteq \Lambda$ with moderate surface
    \be\label{H1additive}
    \hat{H}_{1i}\approx \hat{H}_{1i}^S \oplus \hat{H}_{1i}^{\Sc} \,.
    \ee
    Suppose \eqref{reaction2} for $\ell=1,\ldots,L$ are the (only) possible reactions (i.e., changes in particle numbers), and \eqref{nuE} holds. Choose any $F$ satisfying \eqref{Fkernel}, and define $\hat{Q}_k$ by \eqref{Qkdef} (so that $[\hat{H},\hat{Q}_k]=0$) and $\hat{Q}_K=\hat{H}$. Then for $\dim\Hilbert_\gmc\gg 1$ and $\mu_{0i}$ satisfying \eqref{fixlambda} and \eqref{munucond}, 
    \be\label{1b3}
    \hat\rho_\gc:= Z_\gc^{-1} \exp\Bigl(-\beta(\hat{H}_0 - \sum_i \mu_{0i} \hat{N}_i) \Bigr)
    \ee
    is locally equivalent to $\hat\rho_\gmc$.}
\end{itemize}

Statement~\hyperref[stat1b]{1b} follows from \hyperref[stat1a]{1a} through the reasoning of Section~\ref{sec:chemical}. The condition \eqref{H1additive} is actually satisfied for reasonable 1-particle Hamiltonians, as we argue in Section~\ref{sec:union} for the Laplacian.

\begin{itemize}
\item [] \textbf{Statement 2a. \label{stat2a}(Ensemble Typicality)} {\it In the setting of Statement~\hyperref[stat1a]{1a}, most pure states $\Psi \in \SSS(\Hilbert_\gmc)$ are such that the reduced density matrix of a macro-small region $S\subset\Lambda$ is approximately a generalized Gibbs density matrix:}
\begin{equation}
    \hat\rho_\Psi^S :=\tr^{\Sc} \vert \Psi \rangle \langle \Psi \vert  \approx \hat\rho_\gG^S\,.
\end{equation}

\item [] \label{stat2b}\textbf{Statement 2b. (Grand-Canonical Typicality)} {\it In the setting of Statement~\hyperref[stat1b]{1b}, most pure states $\Psi \in \SSS(\Hilbert_\gmc)$ are such that the reduced density matrix of a macro-small region $S\subset\Lambda$ is approximately grand-canonical:}
\begin{equation}
    \hat\rho_\Psi^S = \tr^{\Sc} \vert \Psi \rangle \langle \Psi \vert  \approx \hat\rho_\gc^S\,.
\end{equation}
\end{itemize}

For the derivation, an application of a theorem of \cite{PSW_06}, see Section~\ref{sec:derivation2}. ``Most'' refers to the uniform measure over $\SSS(\Hilbert_\gmc)$. The statements show that for $\Psi\in\mathbb{S}(\Hilbert_\gmc)$ typical for the generalized micro-canonical ensemble, the reduced density matrix $\hat{\rho}_\Psi^S$ is close to $\hat{\rho}_\gG^S$ resp. $\hat{\rho}_\gc^S$ (``ensemble typicality'').

\subsection{Distribution of the Wave Function}
\label{sec:3.2}

For the next statement, we need the notion of conditional wave function $\psi^S$ \cite{DGZ_92,GN_99,GLTZ_06}, which we briefly explain now. The concept starts from the wave function $\Psi$ of $S$ and $B$ together and makes use of an orthonormal basis (ONB) $b=\{|b_1\rangle,|b_2\rangle,\ldots\}$ of $\Hilbert^B$ (or generalized ONB, such as the position basis); the conditional wave function is the random wave function
\be\label{psiSdef}
\psi^S=\mathcal{N} \langle b_J |\Psi\rangle_B\in\SSS(\Hilbert^S)\,,
\ee
where $\mathcal{N}$ is the normalizing factor, $\langle \cdot | \cdot \rangle_B$ is the partial inner product, and the basis vector $b_J$ is chosen randomly with Born distribution,
\be\label{biprob}
\PPP(J=j)=\bigl\| \langle b_j|\Psi\rangle_B \bigr\|_S^2\,.
\ee
The conditional wave function can be thought of as what the collapsed wave function of $S$ would be after an ideal quantum measurement of the basis $b$ on $B$ \cite[Footnote 2]{GLMTZ_15}.

In the following, the GAP measures play a role. For every density matrix $\hat\rho$ on a Hilbert space $\Hilbert$, the measure $\GAP_{\hat\rho}$ is a probability measure on the unit sphere of $\Hilbert$ whose associated density matrix is $\hat\rho$; it is the most spread-out measure for the given density matrix $\hat\rho$; see Section~\ref{sec:GAP} for the definition.

We now answer the question, raised in the introduction, of what $\PPP_\gc$ is, the probability distribution of the wave function corresponding to the grand-canonical ensemble. We answer it by providing the typical distribution of the conditional wave function $\psi^S$. The answer is twofold, as it depends on how the basis $b$ is chosen. In general, $\psi^S$ (and its distribution) depends on the choice of $b$, and while $b$ can be chosen arbitrarily, one is often particularly interested in taking $b$ to be the position (or configuration) basis, for several reasons: because that is the basis used in the original concept \cite{DGZ_92} of conditional wave function, because of the reasoning of Footnote~\ref{fn:cond}, and because it is the one directly relevant if Bohmian mechanics is true.\footnote{Bohmian mechanics \cite[e.g.,][]{DGZ_92} is a leading candidate theory for providing clear foundations of quantum mechanics; it does so by postulating that particles have positions, thereby giving the position basis a special status.} 
We found it to be too difficult to determine the distribution of $\psi^S$ for $b$ the position basis, but we have results involving two kinds of typical bases, which can serve as reasonable simplified \emph{models} for the position basis; both assume that $S$ is a macro-small spatial region and $B=\Sc$: one result in which $b$ is a \emph{typical} (i.e., random) ONB, and one in which $b$ is \emph{typical among those that diagonalize the conserved operators $\hat{Q}_k^{\Sc}$} (respectively, the particle number operators $\hat{N}_i^{\Sc}$). The first result (Statement~\hyperref[stat3]{3}) applies whenever $b$ is unrelated to the joint eigenbasis of $\hat{H}^{\Sc}$ and the $\hat{Q}_k^{\Sc}$ (respectively, $\hat{N}_i^{\Sc}$); the second (Statements \hyperref[stat4a]{4a} and \hyperref[stat4b]{4b}) whenever there is no simple relation between $b$ and $\hat{H}^{\Sc}$ other than that both $b$ and a suitable eigenbasis of $\hat{H}^{\Sc}$ diagonalize the $\hat{Q}_k^{\Sc}$ (the $\hat{N}_i^{\Sc}$)---this seems like a reasonable approximation for the position basis.

For any random variable $X$, let $\law_X$ (``law of $X$'') denote its probability distribution.

\begin{itemize}
\item [] \textbf{Statement 3. \label{stat3} ($\PPP_\gG$ and $\PPP_\gc$ for Typical ONB)}
{\it In the setting of Statement~\hyperref[stat1a]{1a} (or \hyperref[stat1b]{1b}) with $S$ a macro-small spatial region, for most $\Psi\in\SSS(\Hilbert_\gmc)$ and most ONBs $b$ of $\Hilbert^{\Sc}$, the conditional wave function \eqref{psiSdef} has  distribution 
\begin{equation}\label{3}
        \law_{\psi^S} \approx \GAP_{\hat{\rho}^S_{\gG}}
\end{equation}
(respectively, $\approx\GAP_{\hat\rho_\gc^S}$).}
\end{itemize}

For the derivation, an application of a theorem of \cite{GLMTZ_15}, see Section~\ref{sec:derivation3}. (The closeness of measures $\mu,\nu$ expressed by $\mu\approx\nu$ could be made precise as meaning that for a suitable class of functions $f:\SSS(\Hilbert)\to\RRR$, the $\mu$-average of $f$ is close to the $\nu$-average, see Sections~\ref{sec:GAP} and \ref{sec:derivation3}.)

\bigskip

In the following, we use the notation
\be\label{muoplusdef}
\mu=\bigoplus_j p(j) \, \mu_j
\ee
for the probability distribution that is the mixture of the probability distributions $\mu_j$ with weights $p(j) \geq 0$, $\sum_j p(j) =1$. Equivalently, if $\law_{X_j}= \mu_j$ and $\PPP(J=j)=p(j)$, then $\law_{X_J} = \mu$. We also use the notation
\be
\hat{P} (\hat{A}=a)
\ee
for the projection to the eigenspace of $\hat{A}$ with eigenvalue $a$ and
\be
\hat{P} \bigl( \hat{A}_1=a_1,\ldots, \hat{A}_r=a_r \bigr)
\ee
for the projection to the joint eigenspace of commuting $\hat{A}_1,\ldots,\hat{A}_r$. By the Born rule, the joint probability distribution of the outcomes in a simultaneous quantum measurement of the observables $\hat{A}_1,\dots,\hat{A}_r$ on a system with density matrix $\hat{\rho}$ is given by
\begin{equation}
    \mathbb{P}\bigl( \hat{A}_1=a_1,\ldots, \hat{A}_r=a_r \bigr)
    = \tr \Bigl[\hat{\rho}\hat{P}\bigl( \hat{A}_1=a_1,\ldots, \hat{A}_r=a_r \bigr)\Bigr].
\end{equation}

\begin{itemize}
\item [] \label{stat4a} \textbf{Statement 4a.  ($\PPP_\gG$ for ONB Diagonalizing $\hat{Q}_k^{\Sc}$)} {\it In the setting of Statement~\hyperref[stat1a]{1a}, suppose that $S\subset\Lambda$ is a macro-small region and that for every $k=1,\ldots,K-1$, $\Delta Q_k=0$ and $Q_k$ is an eigenvalue of $\hat{Q}_k$. Then for most $\Psi\in\SSS(\Hilbert_\gmc)$ and most ONBs $b$ of $\Hilbert^{\Sc}$ that diagonalize $\hat{Q}_1^{\Sc},\ldots,\hat{Q}_{K-1}^{\Sc}$ (but not $\hat{Q}^{\Sc}_K=\hat{H}^{\Sc}$), 
    \begin{equation}\label{4a}
        \law_{\psi^S} \approx \PPP_\gG :=\bigoplus_{q_1^S,\ldots,q_{K-1}^S} p(q_1^S,\ldots,q_{K-1}^S) ~ \GAP_{\hat{\rho}^S(q_1^S,\ldots,q_{K-1}^S)}\,,
    \end{equation}
where
\be\label{pqdef}
p(q_1^S,\ldots,q_{K-1}^S) = \tr(\hat\rho^S_\gG \hat{P})~~\text{with}~~\hat{P}:=\hat{P}\Bigl( \hat{Q}_1^S=q_1^S,\ldots,\hat{Q}_{K-1}^S=q_{K-1}^S \Bigr)
\ee
and}
\be\label{rhoqdef}
\hat{\rho}^S(q_1^S, \ldots, q_{K-1}^S) = \frac{1}{p(q_1^S,\ldots,q_{K-1}^S)} \hat{P}\, \hat\rho_\gG^S \, \hat{P}\,.
\ee

\item [] \textbf{Statement 4b. ($\PPP_\gc$ for ONB Diagonalizing $\hat{N}_i^{\Sc}$)} \label{stat4b} {\it In the setting of Statement~\hyperref[stat1b]{1b}, suppose that $S\subset\Lambda$ is a macro-small region and that for every $k=1,\ldots,K-1$, $\Delta Q_k=0$ and $Q_k$ is an eigenvalue of $\hat{Q}_k$. Then for most $\Psi\in\SSS(\Hilbert_\gmc)$ and most ONBs $b$ of $\Hilbert^{\Sc}$ that diagonalize $\hat{N}_1^{\Sc},\ldots,\hat{N}_r^{\Sc}$, 
\begin{equation}\label{4b}
        \law_{\psi^S} \approx \PPP_\gc := \bigoplus_{n_1^{\Sc},...,n_r^{\Sc}} p_n(n_1^{\Sc},\ldots,n_r^{\Sc}) ~ \GAP_{\hat{\rho}(n_1^{\Sc},\ldots,n_r^{\Sc})}\,,
    \end{equation}
where
\be\label{pndef}
p_n(n_1^{\Sc},\ldots,n_r^{\Sc}) = \tr(\hat\rho_\gmc \hat{P}_n)~~\text{with}~~\hat{P}_n:= \hat{P}\Bigl(\hat{N}_1^{\Sc}=n_1^{\Sc},\ldots,\hat{N}_r^{\Sc}=n_r^{\Sc} \Bigr)
\ee
and}
\be
\hat\rho(n_1^{\Sc}, \ldots, n_r^{\Sc}) = \frac{1}{p_n(n_1^{\Sc},\ldots,n_r^{\Sc})} \tr^{\Sc}\bigl[\hat{P}_n \, \hat\rho_\gmc \, \hat{P}_n \bigr]\,.
\ee
\end{itemize}

For the derivations, see Sections~\ref{sec:der4a} and \ref{sec:der4b}.

\begin{rmk}\label{rem:4b+}
If the general Gibbs principle  is more generally valid than under the assumptions of Statement~\hyperref[stat1a]{1a} (which is plausible as pointed out in Section~\ref{sec:plausibility}), then the expression for $\PPP_\gc$ can be simplified further:
\begin{equation}\label{4b+}
        \PPP_\gc \approx \bigoplus_{q_1^S,...,q_{K-1}^S} p(q_1^S,\ldots,q_{K-1}^S) ~ \GAP_{\hat{\rho}^S(q_1^S,\ldots,q_{K-1}^S)}\,,
\end{equation}
where
\be\label{pdef}
p(q_1^S,\ldots,q_{K-1}^S) = \tr(\hat\rho^S_\gc \hat{P})~~\text{with $\hat{P}$ as in \eqref{pqdef}}
\ee
and
\be\label{rhoSdef}
\hat\rho^S(q_1^S, \ldots, q_{K-1}^S) = \frac{1}{p(q_1^S,\ldots,q_{K-1}^S)} \hat{P} \, \hat\rho_\gc^S \, \hat{P} \,.
\ee
\hfill$\diamond$
\end{rmk}

Statement~\hyperref[stat4b]{4b} is actually not a special case of Statement~\hyperref[stat4a]{4a} because only certain combinations of $\hat{N}_i$ are conserved as the $\hat{Q}_k$, and while in Statement~\hyperref[stat4a]{4a} we considered bases diagonalizing the $\hat{Q}_k^{\Sc}$ (i.e., certain combinations of the $\hat{N}_i^{\Sc}$), we consider in Statement~\hyperref[stat4b]{4b} bases diagonalizing \emph{all} $\hat{N}_i^{\Sc}$. That is because we are interested in the position/configuration basis, which diagonalizes all $\hat{N}_i^{\Sc}$; on the other hand, Statement~\hyperref[stat4a]{4a} is the simpler statement, and can be applicable to general $\hat{Q}_k$ that are not related to number operators.

We also note in passing that while the $\oplus$ symbol for measures means merely, according to its definition \eqref{muoplusdef}, the mixture of several measures, something more is true in \eqref{4a} and \eqref{4b+}: the measures are even disjoint and, on top of that, their supports even lie in mutually orthogonal subspaces.

\subsection{Approach to Equilibrium}

The next statement concerns the approach (time evolution) towards thermal equilibrium. As many earlier works such as \cite{vN29,Rei_08b,LPSW,GLMTZ_10}, we take ``approach'' to mean that for most $t\geq 0$ in the long run, the system is in thermal equilibrium. And we take ``thermal equilibrium'' to mean that $\hat\rho_{\Psi_t}^S \approx \hat\rho_\gG^S$ and $\law_{\psi^S} \approx \PPP_\gG$. 

\begin{itemize}
\item [] \textbf{Statement 5. (Approach to Equilibrium)}\label{stat5} {\it Consider the setting of Statement~\hyperref[stat3]{3} (or \hyperref[stat4a]{4a} or \hyperref[stat4b]{4b}) with macro-small spatial region  $S\subset \Lambda$. Suppose $\hat{H}$ satisfies the eigenstate thermalization hypothesis (ETH): for eigenvectors $\phi_1,\phi_2\in\SSS(\Hilbert_\gmc)$ of $\hat{H}$ that belong to different eigenspaces of $\hat{H}$,
\begin{subequations}
\begin{align}
\tr^{\Sc} |\phi_1\rangle \langle \phi_1| &\approx \tr^{\Sc} \hat\rho_\gmc \label{ETHa}\\
\tr^{\Sc} |\phi_1\rangle \langle \phi_2| &\approx 0\,. \label{ETHb}
\end{align}
\end{subequations}
Then for every $\Psi_0\in\mathbb{S}(\Hilbert_\gmc)$ for most $t\geq 0$, the reduced density matrix of $\Psi_t=e^{-i\hat{H}t}\Psi_0$ is
\be\label{5}
\hat\rho_{\Psi_t}^S=\tr^{\Sc} |\Psi_t\rangle \langle \Psi_t| \approx \hat\rho_\gG^S~~(\text{respectively, }\hat\rho_\gc^S)\,.
\ee
Furthermore, for every $\Psi_0\in\mathbb{S}(\Hilbert_\gmc)$
for most ONBs $b$ as in Statement~\hyperref[stat3]{3} (respectively, \hyperref[stat4a]{4a} or \hyperref[stat4b]{4b}) for most $t\geq 0$, the conditional wave function $\psi^S(t)$ obtained from $\Psi_t$ is (approximately) distributed as in \eqref{3} (respectively, \eqref{4a} or \eqref{4b}).}
\end{itemize}

For the derivation, see Section~\ref{sec:approach}.

The version of ETH expressed in \eqref{ETHa}-\eqref{ETHb}, due to Srednicki \cite{Sre94}, is the one appropriate for \emph{microscopic} thermal equilibrium (``MITE'') \cite{GHLT_17}; for macroscopic thermal equilibrium, a different condition would be relevant \cite{GLMTZ_10,GHLT_17,RSTTV_25}.

\section{Discussion}
\label{sec:discussion}

\subsection{Comparison to Maximum Entropy Principle}
\label{sec:comparison}

The density matrices $\hat\rho_\can$ and $\hat\rho_\gc$ are often introduced in the literature as the maximizers of the von Neumann entropy $S_{\mathrm{vN}}$ \eqref{SvN} under the constraint $\tr(\hat\rho \hat{H})=E$ (and, if appropriate, $\tr(\hat\rho \hat{N})=n$) (e.g., \cite{Bal06,GE_16,Kar07}, classically \cite{Jay_57}). While it is mathematically correct that they \emph{are} these maximizers, a derivation based merely on this fact is considerably weaker than the one we have presented, and in fact rather questionable, for at least two reasons as follows. 

\subsubsection{Differences Between Typicality Reasoning and Maximum Entropy Principle}

First, the maximization problem leads to a representation of \emph{subjective} probability (or the knowledge of an observer), but if statistical mechanics were limited to subjective probability then it would not be justified in most applications; in particular, it could not treat any events that take place in the absence of observers (e.g., the formation of stars before humans even existed).
In contrast, there is nothing subjective about $\tr^{\Sc}|\Psi\rangle \langle \Psi|$; that is, we have shown how $\hat\rho_\gc$ actually occurs in nature.

Second, it remains unclear why one should constrain the \emph{expectation value} of a probability distribution or density matrix, even when determining the subjective probability distribution of an observer with limited information. After all, if a person was looking for their car and somehow knew only that its location has latitude $50^\circ$ N, then the subjective probability distribution would be concentrated on the circle of latitude $50^\circ$ N, rather than including other latitudes in such a fashion that the \emph{expected} latitude is $50^\circ$ N. Therefore, even for subjective probability, $\hat\rho_\mc$ or $\hat\rho_\gmc$ would be much more relevant; it is only after realizing and using the equivalence of ensembles that they can be replaced by $\hat\rho_\can$ or $\hat\rho_\gc$. A precise version of this equivalence is provided by our Statements \hyperref[stat2a]{2a} and \hyperref[stat2b]{2b}. 

\begin{rmk}\label{rmk:Balian}
    The second issue above was mentioned by Balian \cite[Sec.~4.1.2]{Bal06}, who suggested that prescribing the expectation value of an observable $\hat{A}$ made sense when considering \emph{subsystems}, of which there are many, and each may have a different value of $\hat{A}$. However, that seems like mixing up the description and the justification of $\hat{\rho}_\can$. If, as Balian's reasoning suggests, $\hat{\rho}_\can$ can only apply to subsystems $S$, while only $\hat{\rho}_\mc$ can apply to full systems, then $\hat\rho^S_\can$ can only be justified as $\tr^{\Sc}\hat\rho_\mc$, and the maximizer of $S_{\mathrm{vN}}$ for $S$ (under whichever constraints) would be irrelevant unless it agreed with $\tr^{\Sc}\hat\rho_\mc$. That is, the reasoning apparently undercuts the use of $S_\mathrm{vN}$ for justifying $\hat\rho_\can$. For us, the consideration of subsystems also plays an important role, for example through the concept of local equivalence of density matrices, but the logic is different in our case, as we do not use subjective probability.\hfill$\diamond$
\end{rmk}

\subsubsection{Origin of the Uniform Measure in Typicality Reasoning}
\label{sec:root}

Another question that arises in this context is whether $\hat\rho_\mc$ (or $u_{\SSS(\Hilbert_\mc)}$) itself has the status of subjective probability, given that it simply assigns equal weight to states with energies in $[E-\Delta E,E]$. In other words, didn't we put in an assumption saying something like ``all states in $\Hilbert_\mc$ are equally probable''? And isn't such an assumption simply an expression of insufficient knowledge and a method of guessing in the absence of knowledge?

Our answer is: no. The use of $u_{\SSS(\Hilbert_\mc)}$ has a different origin; it has nothing to do with anybody's knowledge, and it doesn't amount to making an \emph{assumption}. Rather, it is a characteristic property of \emph{thermal equilibrium} that the micro-state (here, $\Psi$) behaves like the overwhelming majority of micro-states with the same values of $E$ (and, if appropriate, $Q_1,\ldots, Q_{K-1}$). Just as most phase points on an energy surface of a classical ideal gas have a Maxwellian empirical distribution of velocities, also most wave functions in $\SSS(\Hilbert_\mc)$ (relative to the natural measure $u_{\SSS(\Hilbert_\mc)}$) look like thermal equilibrium \cite{GLMTZ_10,GHLT_17} (with certain exceptions \cite[Sec.~8]{GHLT_17}). Therefore, it can be taken as the \emph{definition} of thermal equilibrium to behave like most wave functions.
A $\Psi$ in thermal equilibrium should for this reason be typical relative to the measure $u_{\SSS(\Hilbert_\mc)}$, whose density matrix is $\hat\rho_\mc$. That is, it is the nature of thermal equilibrium that leads to $u_{\SSS(\Hilbert_\mc)}$ (respectively, $u_{\SSS(\Hilbert_\gmc)}$), and $\hat\rho_\mc$ ($\hat\rho_\gmc$).

\begin{rmk} 
Typicality reasoning has, apart from successful applications in quantum statistical mechanics (see, e.g., \cite{GE_16,Mori_18} for overviews), also gained growing attention in the philosophical literature. For more details on the conceptual meaning of typicality, as well as its application to other fields in physics see \cite{Gol12,LR15,Wil22,Laz23,Wil25}. 
\end{rmk}

\subsection{Plausibility of the General Gibbs Principle}
\label{sec:plausibility}

The reasoning in this section is of a more heuristic character than in other parts of this paper: what we present here is a mere plausibility argument.

``Equivalence of ensembles'' means that one thermodynamic ensemble can be replaced, in either classical or quantum statistical mechanics, by another; many variations of this theme are known, corresponding to different senses of ``equivalence.'' We now describe a heuristic consideration that makes it plausible that $\hat\rho_\gmc$ as in \eqref{rhogmcdef} is equivalent to $\hat\rho_\gG$ as in \eqref{rhogGdef}, without referring to a precise sense of equivalence. (A similar consideration applies already to $\hat\rho_\mc$ and $\hat\rho_\can$.)

Any ONB $\{\phi_j\}$ that jointly diagonalizes $\hat{Q}_1,\ldots,\hat{Q}_K$ also diagonalizes $\hat\rho_\gmc$ and $\hat\rho_\gG$, so the relevant question is whether their \emph{eigenvalues} ``look similar.'' Each of the two density matrices defines a Born distribution in the real space spanned by $K$ axes labeled with $q_1,\ldots,q_K$, given by
\be
p_\alpha(q_1,\ldots,q_K)=\tr\Bigl( \hat\rho_\alpha \, \hat{P}\bigl( \hat{Q}_1=q_1,\ldots, \hat{Q}_K=q_K\bigr)\Bigr)
\ee
with $\alpha=\gmc$ or $\alpha=\gG$,
and the consideration aims to show that the two probability distributions ``look similar'' by showing that both are sharply peaked around the point $(Q_1,\ldots,Q_K)$: $p_\gmc$ because it is concentrated in the rectangle $\prod_k[Q_k-\Delta Q_k,Q_k]$ and $\Delta Q_k$ is small, and $p_\gG$ because its expectation is $(Q_1,\ldots,Q_K)$ by \eqref{fixlambda}, much larger values of $q_k$ are suppressed by the exponential dependence in \eqref{rhogGdef} (provided $\lambda_k<0$), and much smaller values are suppressed provided the joint ``density of states'' of $\hat{Q}_1,\ldots,\hat{Q}_K$ (i.e., the dimension of the spectral subspace associated with a cube around $(q_1,\ldots,q_K)$ with side length $\Delta Q$ that is small on the macroscopic scale, divided by the volume $(\Delta Q)^K$ of the cube and then approximated by a smooth function of $q_1,\ldots, q_K$) increases quickly with each of the $q_k$ variables, as it would be the case for relevant examples of macroscopic observables such as energy or particle number.

It is possible to make this consideration a little bit sharper: A relevant sense of ``equivalence'' could be that 
\be\label{macroAequiv}
\text{for macroscopic observables $\hat{A}$,}~~\tr(\hat\rho_\gmc \hat{A}) \approx \tr(\hat\rho_\gG \hat{A}) \,. 
\ee
The width of the peak will be different for $p_\gmc$ and $p_\gG$:

\begin{ex}\label{ex:spread}
   In the canonical case $K=1$, $\hat{Q}_1=\hat{H}$, $\hat\rho_\gmc=\hat\rho_\mc$, and $\hat\rho_\gG=\hat\rho_\can$, the density of states of $\hat{H}=-\Delta$ is a power law $f(q_1)\propto q_1^M$ with a large exponent $M$ proportional to the particle number, and one finds that $p_\gmc(q_1) f(q_1)$ has spread of order $Q_1/M$ (see Appendix~\ref{app:spread}), while $p_\gG(q_1) f(q_1)$ has spread of order $Q_1/\sqrt{M}$, which is larger by the large factor $\sqrt{M}$.\hfill$\diamond$ 
\end{ex}

To be sure, there are observables, such as
\be
\mathds{1}_{[-1/\sqrt{M},-1/M]\cup[1/M,1/\sqrt{M}]}(\hat{Q}_1/Q_1-\hat{I})\,,
\ee
that yield $\approx 0$ for the narrower peak but $\approx 1$ for the wider peak; however, relevant macroscopic observables $\hat{A}$ should be insensitive to the width of the peak, which supports \eqref{macroAequiv}.

\subsection{Laplacian on the Union of Regions}
\label{sec:union}

Statement~\hyperref[stat1b]{1b} assumed the condition~\eqref{H1additive} that the 1-particle Hamiltonian $\hat{H}_1$ is approximately additive relative to a spatial decomposition in two regions. In this section, we argue that this is a reasonable condition that should be satisfied in practice. To this end, we consider, as a concrete example of $\hat{H}_1$, the Laplacian on the union of two regions.

So consider two spatial regions $\Lambda_1,\Lambda_2\subset \mathbb{R}^3$ bordering on each other along some surface $\Sigma$. We argue that the free 1-particle Hamiltonian associated with $\Lambda=\Lambda_1 \cup \Lambda_2$ is approximately, though not exactly, block diagonal with respect to the corresponding splitting $\Hilbert= \Hilbert_1 \oplus \Hilbert_2$ of the 1-particle Hilbert space, provided the surface $\Sigma$ is moderate (i.e., its area is not large compared to $\mathrm{vol}(\Lambda_i)^{2/3}$; if the space dimension is $d$, then 2/3 should be replaced by $(d-1)/d$).

The argument focuses on the negative Laplacian as an example of a free Hamiltonian $\hat{H}_0$. If $\hat{H}_0$ were exactly block diagonal, then a wave function initial concentrated in $\Lambda_1$ would never enter $\Lambda_2$, but it does. Here is another way of looking at the issue: the negative Laplacian in a region $\Lambda$ becomes a self-adjoint operator only if we introduce boundary conditions on the boundary $\partial\Lambda$ of $\Lambda$; suppose we always use Dirichlet boundary conditions (i.e., $\psi|_{\partial\Lambda}=0$) to define the self-adjoint operator $\Delta_\Lambda$; then it is clear that $(-\Delta_{\Lambda_1})\oplus (-\Delta_{\Lambda_2})$ (which requires $\psi$ to vanish on $\Sigma$) is different from $-\Delta_\Lambda$ (which does not require $\psi$ to vanish on $\Sigma$). Now the argument for being approximately block diagonal, i.e.,
\be
-\Delta_\Lambda \approx (-\Delta_{\Lambda_1})\oplus (-\Delta_{\Lambda_2})\,,
\ee
considers the discrete Laplacian (for simplicity in 1d) on the lattice $\varepsilon \mathbb{Z}$ with mesh width $\varepsilon>0$ and (for simplicity) $\Lambda_1= \{x\in \varepsilon\mathbb{Z}:x\leq 0\}$ and $\Lambda_2 = \{ x\in \varepsilon\mathbb{Z}:x>0\}$ (so the surface is moderate); then it notes that the difference between $-\Delta_\Lambda$ and $(-\Delta_{\Lambda_1})\oplus (-\Delta_{\Lambda_2})$, i.e., between
\be
\frac{1}{\varepsilon^2} \left(
\begin{array}{ccc|ccc}
     \ddots&-1&&&&\\
     -1&2&-1&&&\\
     &-1&2&-1&&\\\hline
     &&-1&2&-1&\\
     &&&-1&2&-1\\
     &&&&-1&\ddots\\
\end{array} 
\right)
~~\text{and}~~
\frac{1}{\varepsilon^2} \left(
\begin{array}{ccc|ccc}
     \ddots&-1&&&&\\
     -1&2&-1&&&\\
     &-1&2&&&\\\hline
     &&&2&-1&\\
     &&&-1&2&-1\\
     &&&&-1&\ddots\\
\end{array} 
\right)
\ee
consists of just two entries---very little.

\section{Derivation of Statement~\hyperref[stat1a]{1a}: Generalized Gibbs Density Matrix}
\label{sec: deriv 1a}

This derivation has parallels with considerations in \cite{Schw_06,LL_80,FG_17,Hua_87,GLTZ_06a}. However, we do not assume here that relations familiar from thermodynamics are valid.

For the purposes of this derivation, we pretend that the approximate equalities in \eqref{Qkextensive} (and thus also \eqref{rhogGSSc} and \eqref{trScrhogG}) are exact equalities. For $S\subset \Lambda$, it follows that
\begin{subequations}
\begin{align}
0&=\bigl[ \hat{Q}_k, \hat{Q}_{k'} \bigr]\\ 
&= \Bigl[ \hat{Q}_k^S \otimes \hat{I}^{\Sc}+ \hat{I}^S\otimes \hat{Q}_k^{\Sc}, \hat{Q}_{k'}^S \otimes \hat{I}^{\Sc} + \hat{I}^S \otimes \hat{Q}_{k'}^{\Sc} \Bigr]\\
&= \bigl[\hat{Q}_k^S,\hat{Q}_{k'}^S \bigr] \otimes \hat{I}^{\Sc} + \hat{I}^S \otimes \bigl[ \hat{Q}_{k}^{\Sc}, \hat{Q}_{k'}^{\Sc} \bigr]\,,
\end{align}
\end{subequations}
so
\be\label{commutatorQkS}
\bigl[\hat{Q}_k^S,\hat{Q}_{k'}^S \bigr] = 0 =  \bigl[ \hat{Q}_{k}^{\Sc}, \hat{Q}_{k'}^{\Sc} \bigr]\,.
\ee
For $X=S,\Sc$ let $\{\phi_j^X:j\in\mathscr{J}^X\}$ be an ONB of $\Hilbert^X$ that is a joint eigenbasis of all $\hat{Q}_k^X$, and let $q_{kj}^X$ be the corresponding eigenvalue. With the notation $\mathscr{J}:=\mathscr{J}^S \times \mathscr{J}^{\Sc}$ and $\phi_{jj'}:= \phi^S_j \otimes \phi_{j'}^{\Sc}$, it follows that $\{\phi_{jj'}: (j,j')\in\mathscr{J}\}$ is an ONB of $\Hilbert$ that diagonalizes (among others) $\hat{Q}_1,\ldots,\hat{Q}_K$. Let
\begin{subequations}
\label{scrJgmcdef}
\begin{align}
\mathscr{J}_\gmc&:=\Bigl\{(j,j')\in\mathscr{J}\::\:\forall k: \hat{Q}_k\phi_{jj'} = q_k\phi_{jj'} \text{ with } q_k \in [Q_k- \Delta Q_k, Q_k]\Bigr\}\\
&= \Bigl\{(j,j') \in \mathscr{J} \::\: \forall k: q_{kj}^S+q_{kj'}^{\Sc} \in [Q_k-\Delta Q_k,Q_k] \Bigr\}\,,
\end{align}
\end{subequations}
so $\mathrm{span}\{\phi_{jj'}: (j,j')\in \mathscr{J}_\gmc \}=\Hilbert_\gmc$. Thus, writing $d_\gmc := \dim \Hilbert_\gmc$,
\be
\hat\rho_\gmc = \frac{1}{d_\gmc} \sum_{(j,j')\in\mathscr{J}_\gmc} |\phi_{jj'}\rangle \langle \phi_{jj'}|
\ee
and
\begin{subequations}
\begin{align}
\tr^{\Sc}\hat\rho_\gmc 
&= \frac{1}{d_\gmc} \sum_{(j,j')\in\mathscr{J}_\gmc} \tr^{\Sc}|\phi_{jj'}\rangle \langle \phi_{jj'}| \\
&= \frac{1}{d_\gmc} \sum_{(j,j')\in\mathscr{J}_\gmc} |\phi_j^S\rangle \langle \phi_j^S|\\
&= \sum_{j\in \mathscr{J}^S} \frac{\#\{j'\in\mathscr{J}^{\Sc}:(j,j')\in \mathscr{J}_\gmc\}}{d_\gmc}|\phi_j^S\rangle \langle \phi_j^S|\\
&= \sum_{j\in \mathscr{J}^S} \frac{\#\bigl\{j'\in\mathscr{J}^{\Sc} : \forall k: q_{kj'}^{\Sc}\in [Q_k-q_{kj}^S-\Delta Q_k,Q_k-q_{kj}^S] \bigr\}}{d_\gmc}|\phi_j^S\rangle \langle \phi_j^S|\\
&= \sum_{j\in \mathscr{J}^S} \frac{\dim \Hilbert_j^{\Sc}}{d_\gmc}|\phi_j^S\rangle \langle \phi_j^S|\,, \label{trrhogmclast}
\end{align}
\end{subequations}
where $\# M$ means the number of elements of the set $M$ and
\be\label{HilbertScProd}
\Hilbert_j^{\Sc} := \prod_{k=1}^K\mathds{1}_{[Q_k-q_{kj}^S-\Delta Q_k,Q_k-q_{kj}^S]}(\hat{Q}_k^{\Sc})  \Hilbert^{\Sc}.
\ee
The Boltzmann entropy of the  macro state of $\Sc$ defined by the values $q_k^{\Sc}$ with tolerance $\Delta Q_k$ is defined by \cite{Gri,GLTZ_20}
\be\label{entropy}
S(q^{\Sc}_1,\ldots,q^{\Sc}_K) := k_\mathrm{B} \log \dim \Biggl( \prod_{k=1}^K \mathds{1}_{[q^{\Sc}_k-\Delta Q_k, q^{\Sc}_k]}(\hat{Q}_k^{\Sc}) \Hilbert^{\Sc} \Biggr)\,,
\ee
so that
\be\label{dimHj}
\dim \Hilbert_j^{\Sc} = \exp\Biggl(\frac{1}{k_\mathrm{B}} S\Bigl(Q_1-q_{1j}^S,\ldots, Q_K-q_{Kj}^S \Bigr) \Biggr).
\ee

We now approximate the $S$ function in \eqref{dimHj} by its Taylor expansion around $(Q_1,\ldots,Q_K)$. If the $\hat{Q}_k, \hat{Q}_k^S$ are positive operators, this can be justified as follows: Given that $S\subset \Lambda$ is macro-small, $\Sc$ is large compared to $S$, and positive extensive observables should be dominated by the contribution from $\Sc$, so that, for most $(j,j')\in\mathscr{J}_\gmc$ and all $k$,
\be
q_{kj}^S \ll q_{kj'}^{\Sc} \approx q_{kj}^S+q_{kj'}^{\Sc} \approx Q_k \,.
\ee
For $\hat{Q}_k$ that is not positive but merely (as we assumed in Statement~\hyperref[stat1a]{1a}) bounded from below, say by $-C_k$, shifting by $C_k$ would justify the Taylor expansion, which then also applies without the shift. 
We thus obtain, as the first-order Taylor expansion,
\be\label{Taylor}
S(Q_1-q_{1j}^S, \dots, Q_K-q_{Kj}^S) \approx S(Q_1,\dots,Q_K) - \sum_{k=1}^{K} \frac{\partial S}{\partial q^{\Sc}_k}(Q_1,\ldots,Q_K) \,  q_{kj}^S.
\ee 
Inserting into \eqref{dimHj} yields
\be
\dim\Hilbert_j^{\Sc} \approx \exp\bigl(S(Q_1,\ldots,Q_K)/k_\mathrm{B} \bigr)\: \exp\left( \sum_{k=1}^K \lambda_k q_{kj}^S\right)
\ee
with
\be\label{lambdaS}
\lambda_k := -\frac{1}{k_\mathrm{B}}\frac{\partial S}{\partial q_k^{\Sc}}(Q_1,\ldots,Q_K)\,.
\ee

Inserting in \eqref{trrhogmclast} and applying functional calculus of operators,
\be 
\tr^{\Sc}\hat\rho_\gmc \approx   \frac{1}{Z_\gG} \exp \left( \sum_{k=1}^{K} \lambda_k \hat{Q}^S_k \right),
\ee 
which is the definition of $\hat{\rho}_\gG^S$. 

For the next steps, we need the following mathematical fact, proved in Appendix~\ref{app:Qb}.

\begin{prop}\label{prop:Qb}
    Let $\Hilbert^a,\Hilbert^b,\Hilbert^c$ be finite-dimensional Hilbert spaces. If for some operators
    \be\label{Qb}
    \hat{Q}^a \otimes \hat{I}^b \otimes \hat{I}^c + \hat{I}^a \otimes \hat{Q}^{bc} = \hat{Q}^{ab}\otimes \hat{I}^c + \hat{I}^a\otimes \hat{I}^b \otimes \hat{Q}^c\,,
    \ee
    then there exists a unique operator $\hat{Q}^b$ such that
    \begin{subequations}
        \begin{align}
            \hat{Q}^{bc} &= \hat{Q}^b \otimes \hat{I}^c + \hat{I}^b \otimes \hat{Q}^c \label{Qbc}\\
            \hat{Q}^{ab} &= \hat{Q}^a \otimes \hat{I}^b + \hat{I}^a \otimes \hat{Q}^b\,.\label{Qab}
        \end{align}
    \end{subequations}
\end{prop}

In the remainder of this section, we write again approximate equalities when making use of the approximate extensivity of the $\hat{Q}_k$ so that in the computations below it is clear which equalities are really exact and which are not.

As a consequence of the extensivity of the $\hat{Q}_k$, if $S_1\subseteq S_2\subseteq \Lambda$ and $S_2\setminus S_1$ all are large compared to the size of a molecule and have moderate surfaces, then from \eqref{Qkextensive} for $S_1$ and $S_2$ we obtain that
\begin{align}
\hat{Q}_k &\approx \hat{Q}_k^{S_1}\otimes \hat{I}^{S_2\setminus S_1} \otimes \hat{I}^{\Sc_2} \nonumber\\
&+ \hat{I}^{S_1} \otimes \hat{Q}_k^{S_2\setminus S_1} \otimes \hat{I}^{\Sc_2} \nonumber\\
&+ \hat{I}^{S_1} \otimes \hat{I}^{S_2\setminus S_1} \otimes \hat{Q}_k^{\Sc_2} \,,
\end{align}
and (by repeated application of Proposition~\ref{prop:Qb}) similarly for any finite partition $\Lambda=\cup_\alpha S_\alpha$ ($S_\alpha\cap S_{\alpha'}=\emptyset$ for $\alpha \neq \alpha'$) into sets that are large compared to the size of a molecule and have moderate surfaces: 
\be\label{Qkalpha}
\hat{Q}_k \approx \sum_\alpha \hat{Q}_k^{S_\alpha} \otimes \hat{I}^{\Sc_\alpha}\,.
\ee

\bigskip

For the derivation of Statement~\hyperref[stat1a]{1a}, it
remains to verify that $\lambda_k$ given by \eqref{lambdaS} agrees (approximately) with $\lambda_k$ solving \eqref{fixlambda}. Starting from \eqref{lambdaS}, subdividing the available volume $\Lambda$ into many macro-small regions $S_\alpha$ and using \eqref{Qkalpha}, we find that
\begin{subequations}
\begin{align}
\tr(\hat\rho_\gG \, \hat{Q}_k)
&\approx \sum_\alpha \tr\Bigl[\hat\rho_\gG \, (\hat{Q}_k^{S_\alpha}\otimes \hat{I}^{\Sc_\alpha}) \Bigr]\\ 
&= \sum_\alpha \tr\Bigl[\tr^{\Sc_\alpha}(\hat\rho_\gG) \, \hat{Q}_k^{S_\alpha} \Bigr]\\ 
&\approx \sum_\alpha \tr\Bigl[\tr^{\Sc_\alpha}(\hat\rho_\gmc) \, \hat{Q}_k^{S_\alpha} \Bigr]\\ 
&= \sum_\alpha \tr\Bigl[ \hat\rho_\gmc \, (\hat{Q}_k^{S_\alpha}\otimes \hat{I}^{\Sc_\alpha}) \Bigr]\\ 
&= \tr(\hat\rho_\gmc \hat{Q}_k)\\[2mm]
&\approx Q_k\,.
\end{align}
\end{subequations}
Thus, the $\lambda_k$ given by \eqref{lambdaS} satisfy \eqref{fixlambda}. This completes the derivation of Statement~\hyperref[stat1a]{1a}.

\section{Derivation of Statements \hyperref[stat2a]{2a} and \hyperref[stat2b]{2b}: Grand-Canonical Typicality}
\label{sec:derivation2}

Statement~\hyperref[stat2b]{2b} follows from \hyperref[stat1b]{1b} and \hyperref[stat2a]{2a} together, so we focus on the derivation of \hyperref[stat2a]{2a}. It is based on a theorem of Popescu, Short and Winter \cite{PSW_06}, often applied to the micro-canonical subspace but proved there for a general subspace\footnote{Indeed, as they put it in \cite{PSW_06}, p.3: ``Furthermore our principle will apply to arbitrary restrictions $R$ that have nothing to do with energy, which may lead to many interesting insights.'' As parts of our paper show, this is the case.} $\Hilbert_R$: {\it for any $\Hilbert^S,\Hilbert^B$ of finite dimension, any high-dimensional subspace $\Hilbert_R \subseteq \Hilbert^S \otimes \Hilbert^B$, and most $\Psi \in \SSS(\Hilbert_R)$,
\be\label{rhoR}
\tr^B |\Psi\rangle \langle \Psi| \approx \tr^B \hat\rho_R
\ee
with $\hat\rho_R = (\dim \Hilbert_R)^{-1} \hat{P}_R$ the normalized projection to $\Hilbert_R$.} 

(The precise statement is that for every $\eta>0$,
\be
u_{\SSS(\Hilbert_R)} \left\{ \Psi \in \SSS(\Hilbert_R): \Bigl\| \tr^B |\Psi\rangle \langle \Psi| - \tr^B \hat\rho_R \Bigr\|_{\tr} \geq \eta+\frac{d^S}{\sqrt{d_R}} \right\} \leq 4\exp \left( -\frac{\eta^2 d_R}{18\pi^3} \right)
\ee
with $d^S=\dim \Hilbert^S$,  $d_R=\dim\Hilbert_R$, and $\| \hat{M} \|_{\tr} = \tr |\hat{M}| = \tr \sqrt{\hat{M}^{*}\hat{M}}$ the trace norm of the operator $\hat{M}$.)

Applying this to $\Hilbert_R:=\Hilbert_\gmc$ with $B=\Sc$ and using that $\dim\Hilbert_\gmc$ is large, we obtain that for most $\Psi\in\SSS(\Hilbert_\gmc)$, 
\be 
 \tr^{\Sc} | \Psi \rangle \langle \Psi | \approx \tr^{\Sc} \hat{\rho}_\gmc \,.
\ee 
By Statement~\hyperref[stat1a]{1a}, $\tr^{\Sc} \hat{\rho}_\gmc \approx \hat{\rho}^S_\gG$, which yields the desired result.

\section{Derivation of Statements \hyperref[stat3]{3}, \hyperref[stat4a]{4a}, \hyperref[stat4b]{4b}: Conditional Wave Function and GAP}

Before deriving Statements \hyperref[stat3]{3}--\hyperref[stat5]{5}, we review the properties of the GAP measure (for 
a pedagogical introduction see \cite{Tum19}).

\subsection{GAP Measure}
\label{sec:GAP}

GAP measures can arise in several ways \cite{GLTZ_06,GLMTZ_15}; the one most relevant to us is as the asymptotic distribution of the conditional wave function $\psi^S$ in the setting of canonical typicality (leading to \eqref{PcanGAP}) and, of course, according to Statements \hyperref[stat3]{3}--\hyperref[stat4b]{4b}. That is, GAP measures are relevant as the distribution of wave functions for Gibbs ensembles.

For any Hilbert space $\Hilbert$ and probability measure $\mu$ on $\SSS(\Hilbert)$, the density matrix of $\mu$ is
\be 
 \hat{\rho}_\mu = \int_{\SSS(\Hilbert)} \mu(d\psi) \: \vert \psi \rangle \langle \psi \vert.
\ee
Note that the map $\mu\mapsto \hat{\rho}_\mu$ is many-to-one; there are several measures leading to the same density matrix. For any $\hat\rho$, $\GAP_{\hat\rho}$ is a particular measure with density matrix $\hat\rho$, i.e.,
\begin{equation}
    \hat{\rho}_{\GAP_{\hat{\rho}}} = \hat{\rho}.
\end{equation}
In fact, $\GAP_{\hat\rho}$ is the most spread-out measure with given $\hat\rho$ \cite{JRW_94}.
The name stems from a procedure for constructing it: G stands for Gaussian, A for adjusted and P for projected. Here is the procedure:

We write $X\sim p$ to express that the random variable $X$ has probability distribution $p$. 
Let $\hat{\rho}=\sum_m p_m |m\rangle\langle m|$ be a density matrix on a finite-dimensional Hilbert space $\Hilbert$ with eigenvalues $p_m$ and corresponding eigen-ONB $\{|m\rangle\}$. Moreover, let $(Z_m)$ be a sequence of independent complex-valued Gaussian random variables\footnote{Recall that $Z$ is a complex-valued Gaussian random variable with mean $\mu$ and variance $\sigma^2$ if and only if $\mbox{Re}Z$ and $\mbox{Im}Z$ are independent and $\mbox{Re} Z \sim\mathcal{N}(\mbox{Re}\mu,\sigma^2/2), \mbox{Im} Z\sim\mathcal{N}(\mbox{Im}\mu, \sigma^2/2)$.} with mean 0 and variances
\begin{align}
    \mathbb{E} |Z_m|^2 = p_m.
\end{align}
Then $\mathrm{G}_{\hat{\rho}}$ is the distribution of the random vector
\begin{align}
    \Psi^\mathrm{G} := \sum_m Z_m |m\rangle.
\end{align}
Note that while $\mathrm{G}_{\hat{\rho}}$ has density matrix $\hat{\rho}$, it is not a distribution on the sphere $\mathbb{S}(\Hilbert)$. However, as the eigenvalues $p_m$ sum up to 1, we immediately see that $\mathbb{E}\|\Psi^\mathrm{G}\|^2 = 1$.

If we now projected the Gaussian measure $\mathrm{G}_{\hat{\rho}}$ to the sphere $\mathbb{S}(\Hilbert)$, the resulting distribution would in general not have density matrix $\hat{\rho}$. In order to obtain a measure with the desired density matrix after projecting to the sphere, we have to adjust the density of $\mathrm{G}_{\hat{\rho}}$. More precisely, we define the Gaussian adjusted measure $\mbox{GA}_{\hat{\rho}}$ on $\Hilbert$ by
\begin{align}
    \mbox{GA}_{\hat{\rho}}(d\psi) = \|\psi\|^2 \mathrm{G}_{\hat{\rho}}(d\psi).
\end{align}
We remark that because of $\mathbb{E}\|\Psi^\mathrm{G}\|^2 = 1$, $\mbox{GA}_{\hat{\rho}}$ defines indeed a probability distribution on $\Hilbert$. 

In the last step of the construction we project $\mbox{GA}_{\hat{\rho}}$ to $\mathbb{S}(\Hilbert)$. Let $\Psi^{\mathrm{GA}}\sim\mbox{GA}_{\hat{\rho}}$. Then $\GAP_{\hat{\rho}}$ is the distribution of the random vector
\begin{align}
    \Psi^{\GAP} := \frac{\Psi^{\mathrm{GA}}}{\|\Psi^{\mathrm{GA}}\|}.
\end{align}
Note that $\Psi^{\GAP}$ is well-defined as $\Psi^{\mathrm{GA}}\neq 0$ almost surely. One can easily see that indeed $\hat{\rho}_{\GAP_{\hat{\rho}}}=\hat{\rho}$ and thus this concludes the construction of GAP measures in the finite-dimensional setting. If $\Hilbert$ is infinite-dimensional, a similar construction is possible, see \cite{Tum_20} for the details. For other constructions of $\GAP_{\hat\rho}$, see \cite{GLMTZ_15}.

Another property we use is the continuous dependence of $\GAP_{\hat\rho}$ on $\hat\rho$, which means that for two density matrices $\hat\rho,\hat\Omega$,
\be \label{GAPcont}
\text{if }\hat{\rho}\approx \hat{\Omega},\text{ then }\GAP_{\hat{\rho}}\approx \GAP_{\hat{\Omega}}\,.
\ee
We remark in passing that a precise version of this relation is known and provided by Lemma~5 in \cite{GLMTZ_15}, which says that for every $0<\varepsilon<1$, every finite-dimensional $\Hilbert$, and every continuous function $f: \SSSH \rightarrow \mathbb{R}$, there is $r=r(\varepsilon,\dim\Hilbert,f)>0$ such that for all $\hat{\rho}, \hat{\Omega}\in \mathscr{D}(\Hilbert)$
\be
\text{if} \quad || \hat{\rho} - \hat{\Omega} ||_{\tr}<r, \text{ then} \quad  |\GAP_{\hat{\rho}}(f) - \GAP_{\hat{\Omega}}(f)| < \varepsilon.
\ee
Here, we used a standard way of expressing the closeness of two measures $\mu,\nu$ on a set $\Omega$ by saying that for any  function $f:\Omega\to\RRR$ (serving as a ``test function''), the $\mu$-average of $f$
\be\label{muf}
\mu(f) :=\int_\Omega \mu(d\omega) \, f(\omega)
\ee
is close to $\nu(f)$. Different classes of test functions (e.g., bounded, continuous) define different types of closeness. 

Further investigation and discussion of GAP measures can be found in \cite[Sec.~3]{GLTZ_06}, \cite{TZ_05,Rei_08,GLMTZ_15,TTV_24,Vog_24}.

\subsection{Derivation of Statement \hyperref[stat3]{3}}
\label{sec:derivation3}

The key ingredient is Theorem 2 from \cite{GLMTZ_15}, which says that {\it for any finite-dimensional $\Hilbert^S,\Hilbert^B$ with large $\dim\Hilbert^B$ and any $\Psi\in\SSS(\Hilbert^S\otimes \Hilbert^B)$, most ONBs $b$ of $\Hilbert^B$ are such that}
\be\label{TypGAP}
\law_{\psi^S} \approx \GAP_{\tr^B |\Psi\rangle \langle\Psi|} \,.
\ee
(More precisely, it says that for every $\varepsilon>0$, $d^B:=\dim\Hilbert^B \geq \max\{4,\dim\Hilbert^S\}$, $\Psi\in\SSS(\Hilbert^S\otimes \Hilbert^B)$, and bounded measurable test function $f:\SSS(\Hilbert^S)\to\RRR$,
\be\label{TypGAP2}
u_{ONB} \Bigl\{b : \bigl|\law_{\psi^S}(f) - \GAP_{\tr^B |\Psi\rangle\langle\Psi|}(f)\bigr|< \varepsilon \, \|f\|_\infty \Bigr\} \geq 1-\frac{4}{\varepsilon^2 d^B}
\ee
with $u_{ONB}$ the uniform distribution over the set of ONBs of $\Hilbert^B$ and $\|\cdot\|_\infty$ the supremum norm.)

We apply this to $B=\Sc$. By Statement~\hyperref[stat2a]{2a}, for most $\Psi\in\SSS(\Hilbert_\gmc)$, $\tr^B |\Psi\rangle \langle\Psi| \approx \hat\rho_\gG^S$. By \eqref{GAPcont},\footnote{Theorems 3 and 4 in \cite{GLMTZ_15} provide precise versions of this reasoning.} $\GAP_{\tr^B |\Psi\rangle \langle\Psi|} \approx \GAP_{\hat\rho_\gG^S}$. Thus, $\law_{\psi^S} \approx \GAP_{\hat\rho_\gG^S}$ as claimed (and likewise in the setting of Statement~\hyperref[stat2b]{2b}).

\begin{rmk}\label{rmk:exactGAP}
Another reasoning leads to the related result that $\law_{\psi^S}$, when \emph{averaged} over $\Psi\in\SSS(\Hilbert_\gmc)$, is close to $\GAP_{\hat\rho_\gG^S}$. 
It is based on the property \cite[Sec.~3, Property 3]{GLTZ_06} that ``GAP is hereditary,'' i.e., that if $\Psi\in\SSS(\Hilbert^S \otimes \Hilbert^B)$ has distribution $\GAP_{\hat\rho^S \otimes \hat\rho^B}$, then for any fixed ONB $b$ of $\Hilbert^B$, $\EEE_\Psi \law_{\psi^S} = \GAP_{\hat\rho^S}$. The difference to \eqref{TypGAP} is twofold: First, the equality is exact, not approximate; and second, it concerns the \emph{average} over $\Psi$, while \eqref{TypGAP} is the case for \emph{most} $\Psi$. 
For our application, the exactness does not help as it gets lost in other steps. The application here looks as follows: Suppose that $\hat\rho_\gmc \approx \hat\rho_\gG$. Then the typicality measure for $\Psi$ is $u_{\SSS(\Hilbert_\gmc)}= \GAP_{\hat\rho_\gmc} \approx \GAP_{\hat\rho_\gG} \approx \GAP_{\hat\rho_\gG^S \otimes \hat\rho_\gG^{\Sc}}$ by \eqref{rhogGSSc} and \eqref{GAPcont}. By hereditarity, for any ONB $b$ of $\Hilbert^{\Sc}$, $\EEE_\Psi \law_{\psi^S}\approx \GAP_{\hat\rho_\gG^S}$.\hfill$\diamond$
\end{rmk}

\subsection{Derivation of Statement \hyperref[stat4a]{4a}}
\label{sec:der4a}

An ONB diagonalizing certain commuting observables $\hat{A}_1,\ldots,\hat{A}_n$ (such as $\hat{Q}_1^{\Sc},\ldots,\hat{Q}_{K-1}^{\Sc}$ or later $\hat{N}_1^{\Sc},\ldots,\hat{N}_r^{\Sc}$) consists of an ONB for each of their joint eigenspaces (which in the relevant cases will have high dimension). Thus, a purely random ONB diagonalizing $\hat{A}_1,\ldots,\hat{A}_n$ consists of a purely random ONB in each of the joint eigenspaces. Since the construction of $\psi^S$ involves picking a random basis vector $b_J$, and since each basis vector lies in one of the joint eigenspaces, we ask what is the probability that the chosen basis vector lies in a particular joint eigenspace (say, the one with eigenvalues $a_1,\ldots,a_n$, denoted $\Hilbert_{a_1...a_n}$). Since
\begin{subequations}
\begin{align}
\PPP(J=j) 
&=\bigl\| \langle b_j|\Psi\rangle_B \bigr\|_S^2\\
&= \bigl\langle \Psi \big| (\hat{I}^S \otimes |b_j\rangle\langle b_j|) \big| \Psi\bigr\rangle_{S\cup B}\,,
\end{align}
\end{subequations}
we find that, given a basis $b$,
\begin{subequations}
\begin{align}
\PPP\Bigl(b_J \in \Hilbert_{a_1...a_n} \Big| b \Bigr) 
&= \sum_{j:b_j \in \Hilbert_{a_1...a_n}} \PPP(J=j)\\
&= \Bigl\langle \Psi \Big|\Bigl(\hat{I}^S \otimes \!\!\! \sum_{j:b_j \in \Hilbert_{a_1...a_n}} \!\!\! |b_j\rangle\langle b_j| \Bigr)\Big| \Psi\Bigr\rangle\\
&= \langle\Psi| (\hat{I}^S \otimes \hat{P}_{a_1...a_n})| \Psi\rangle=:\tilde{p}(a_1,\ldots,a_n)\,,\label{Bornalpha}
\end{align}
\end{subequations}
which is independent of the choice of basis $b$, as it depends only on the projection $\hat{P}_{a_1...a_n}$ to the joint eigenspace.
As a consequence, the random variable $\psi^S$ can be constructed from a given $\Psi$ as follows: instead of first choosing a basis $b$ diagonalizing $\hat{A}_1,\ldots,\hat{A}_n$ in all of $\Hilbert^B$ and then a random $J$ with Born distribution, we can just as well choose first a joint eigenspace $\Hilbert_{a_1...a_n}$ with Born distribution \eqref{Bornalpha}, then a random ONB $b'$ in this subspace, and then one basis vector $b'_{J'}$ with conditional Born distribution, given that the eigenspace was $\Hilbert_{a_1...a_n}$. Finally, we can form
\be
\psi^S=\mathcal{N} \, \langle b'_{J'}|\Psi\rangle \,.
\ee
This conditional Born distribution is exactly the Born distribution associated with
\be
\Psi'= \frac{(\hat{I}^S \otimes \hat{P}_{a_1...a_n})\Psi}{\|(\hat{I}^S \otimes \hat{P}_{a_1...a_n}) \Psi\|} \,,
\ee
which can be thought of as the collapsed wave function after a simultaneous quantum measurement of $\hat{A}_1,\ldots,\hat{A}_n$ on $\Psi$ with outcomes $a_1,\ldots,a_n$. Since $b'_{J'}\in \Hilbert_{a_1...a_n}$, $\psi^S$ can just as well be obtained from $\Psi'$ instead of $\Psi$,
\be
\psi^S=\mathcal{N}' \, \langle b'_{J'}|\Psi'\rangle 
\ee
with different normalizing constant $\mathcal{N}'$. That is, after choosing $a_1,\ldots,a_n$ with Born distribution, we can simply continue with $\Psi'$, so that $\Hilbert^B$ gets replaced by $\Hilbert_{a_1...a_n}$. In particular, Theorem 2 of \cite{GLMTZ_15} (as described around \eqref{TypGAP} above) can be applied as long as $\dim\Hilbert_{a_1...a_n}\gg 1$, and yields that for most ONBs $b'$ of $\Hilbert_{a_1...a_n}$, the distribution of $\psi^S$ is approximately $\GAP_{\tr^B |\Psi'\rangle\langle \Psi'|}$. Thus, including all possible values of $a_1,\ldots,a_n$,
\be\label{lawpsiPsi'}
\law_{\psi^S} \approx \bigoplus_{a_1...a_n} \tilde{p}(a_1,\ldots,a_n) ~\GAP_{\tr^B |\Psi'\rangle\langle \Psi'|}\,.
\ee

We now apply this to the setting of Statement~\hyperref[stat4a]{4a}. Since
\begin{subequations}
\begin{align}
\Psi &\in \Hilbert_\gmc \subset \widetilde{\Hilbert}:= \prod_{k=1}^{K-1} \mathds{1}_{\{Q_k\}}(\hat{Q}_k) \Hilbert\\ 
&= \bigoplus_{q_1^S...q_{K-1}^S} \Biggl(\prod_{k=1}^{K-1} \mathds{1}_{\{q_k^S\}}(\hat{Q}_k^S) \Hilbert^S \Biggr) \otimes \Biggl(\prod_{k=1}^{K-1}\mathds{1}_{\{Q_k-q_k^S\}}(\hat{Q}_k^{\Sc}) \Hilbert^{\Sc} \Biggr)  \,,
\end{align}
\end{subequations}
we have that
\be\label{111}
\Bigl( \hat{P}(\hat{Q}_1^S=q_1^S,\ldots) \otimes \hat{I}^{\Sc}\Bigr)\Big|_{\widetilde{\Hilbert}} = 
\Bigl(\hat{I}^S\otimes\hat{P}(\hat{Q}_1^{\Sc}=Q_1-q_1^S,\ldots)\Bigr)\Big|_{\widetilde{\Hilbert}} \,.
\ee
Thus,
\be
\tilde{p}(Q_1-q_1^S,\ldots)=p(q_1^S,\ldots)
\ee
as defined in \eqref{pqdef}, and (taking $\hat{P}(\cdots)$ to mean operators on $\Hilbert$ or $\Hilbert^S$ or $\Hilbert^{\Sc}$ depending on the context)
\begin{subequations}
\begin{align}
\tr^{\Sc} |\Psi'\rangle \langle \Psi'| 
&= (\mathcal{N}')^2 \tr^{\Sc}\Biggl(\hat{P}(\hat{Q}_1^{\Sc}=Q_1-q_1^S,\ldots)~ |\Psi\rangle \langle \Psi|~\hat{P}(\hat{Q}_1^{\Sc}=Q_1-q_1^S,\ldots) \Biggr)\\
&=(\mathcal{N}')^2 \tr^{\Sc}\Biggl(\hat{P}(\hat{Q}_1^S=q_1^S,\ldots)~ |\Psi\rangle \langle \Psi|~\hat{P}(\hat{Q}_1^S=q_1^S,\ldots) \Biggr)\\
&=(\mathcal{N}')^2 \hat{P}(\hat{Q}_1^S=q_1^S,\ldots) ~\tr^{\Sc}\Bigl( |\Psi\rangle \langle \Psi| \Bigr)~\hat{P}(\hat{Q}_1^S=q_1^S,\ldots)\\
&\approx(\mathcal{N}')^2 \hat{P}(\hat{Q}_1^S=q_1^S,\ldots) ~\hat\rho_\gG^S~\hat{P}(\hat{Q}_1^S=q_1^S,\ldots) = \eqref{rhoqdef}
\end{align}
\end{subequations}
for most $\Psi\in\SSS(\Hilbert_\gmc)$ by Statement~\hyperref[stat2a]{2a}. By \eqref{GAPcont}, we obtain Statement~\hyperref[stat4a]{4a}.

\subsection{Derivation of Statement \hyperref[stat4b]{4b}}
\label{sec:der4b}

We apply \eqref{lawpsiPsi'} to $n=r$ and $\hat{A}_i=\hat{N}_i^{\Sc}$. By Theorem 1 in \cite{TTV_22} for $t=0$, for most $\Psi\in\SSS(\Hilbert_\gmc)$, $\tilde{p}(n_1^{\Sc},\ldots,n_r^{\Sc})\approx \tr(\hat\rho_\gmc \hat{P}_n)=p_n(n_1^{\Sc},\ldots,n_r^{\Sc})$. Given $n_1^{\Sc},\ldots,n_r^{\Sc}$, $\Psi'$ is GAP-distributed with density matrix $\hat{P}_n \hat\rho_\gmc \hat{P}_n/p_n(n_1^{\Sc},\ldots,n_r^{\Sc})$. (This can be shown.) By the results of \cite{TTV_24}, $\tr^{\Sc} |\Psi'\rangle \langle \Psi'| \approx \hat\rho(n_1^{\Sc},\ldots,n_r^{\Sc})$ with probability close to 1. Thus, \eqref{lawpsiPsi'} takes the form \eqref{4b}. This completes the derivation of Statement~\hyperref[stat4b]{4b}.

\bigskip

Remark~\ref{rem:4b+} can now be obtained as follows: When acting on $\Hilbert_\gmc$, $\hat{I}^S \otimes \hat{P}_n$ can be replaced by $\hat{P} \otimes \hat{P}_n$ with $\hat{P}$ as in \eqref{pqdef} and $q_k^S := Q_k-F_k(n_1^{\Sc},\ldots,n_r^{\Sc})$. Given that we can otherwise replace $\hat\rho_\gmc$ by $\hat\rho_\gG \approx \hat\rho_\gc$, and using \eqref{rhogGSSc}, we find that
\begin{subequations}
\begin{align}
p_n(n_1^{\Sc},\ldots,n_r^{\Sc})
&= \tr\bigl[\hat\rho_\gmc
(\hat{I}^S \otimes \hat{P}_n)\bigr]\\ 
&= \tr\bigl[\hat\rho_\gmc
(\hat{P} \otimes \hat{P}_n)\bigr]\\ 
&\approx \tr\bigl[ \hat\rho_\gG (\hat{P} \otimes \hat{P}_n)\bigr] \\
&=\tr(\hat\rho_\gG^S \hat{P}) \, \tr(\hat\rho_\gG^{\Sc}\hat{P}_n)
\end{align}
\end{subequations}
and in the same way $\hat\rho(n_1^{\Sc},\ldots,n_r^{\Sc})\approx  \hat\rho^S(q_1^S,\ldots,q_{K-1}^S)$ as in \eqref{rhoqdef}. Thus, all summands in \eqref{4b} with equal values of $q_1^S,\ldots,q_{K-1}^S$ are multiples of the same distribution $\GAP_{\hat\rho^S(q_1^S...q_{K-1}^S)}$, and thus can be combined into a single summand with weight
\be
\sum_{\substack{n_1^{\Sc}...n_r^{\Sc}:\\
Q_k-F_k(n_1^{\Sc},\ldots,n_r^{\Sc})=q_k^S}} p_n(n_1^{\Sc},\ldots, n_r^{\Sc}) \approx \tr(\hat\rho_\gG^S \hat{P})= p(q_1^S,\ldots,q_{K-1}^S)\,,
\ee
which leads to the equations of Remark~\ref{rem:4b+}.

\begin{ex} 
As a simple example we consider the case of one particle species $(r=1)$. In this situation there is no notion of chemical equilibrium but only the one of spatial equilibrium.

Let $\Hilbert = \Hilbert^S \otimes \Hilbert^{\Sc}$ and let $\hat{Q}_1 = \hat{N}$ be the particle number operator of the total system. Obviously, $\hat{N}$ is (exactly) extensive and thus we can decompose it as
\begin{align}
\hat{N} = \hat{N}^S \otimes \hat{I}^{\Sc} + \hat{I}^S \otimes \hat{N}^{\Sc},
\end{align}
where $\hat{N}^S$ and $\hat{N}^{\Sc}$ are the particle number operators of $S$ and $\Sc$, with eigenvalues $n^S$ and $n^{\Sc}$, acting on $\Hilbert^S$ and $\Hilbert^{\Sc}$ respectively. 

Let $\Delta N = 0$ and let $N\in\mathbb{N}$ be an eigenvalue of $\hat{N}$, so $\Hilbert_\gmc$ is characterized by the (exact) total particle number and a micro-canonical energy interval.
By Statement~\hyperref[4b]{4b}, for most $\Psi\in\mathbb{S}(\Hilbert_{\gmc})$ and most ONBs $b$ of $\Hilbert^{\Sc}$ that diagonalize $\hat{N}^{\Sc}$,
\begin{align}
	\law_{\psi^S} \approx \bigoplus_{n^{\Sc}=0}^N p_n(n^{\Sc}) \, \GAP_{\hat{\rho}(n^{\Sc})},\label{eq: law ex}
\end{align}
where
\begin{align}
	p_n(n^{\Sc}) = \tr(\hat{\rho}_\gmc \hat{P}_n), \qquad \hat{\rho}(n^{\Sc}) = \frac{1}{p_n(n^{\Sc})} \tr^{\Sc}\left[\hat{P}_n \hat{\rho}_\gmc \hat{P}_n\right],
\end{align}
with $\hat{P}_n=\hat{P}(\hat{N}^{\Sc}=n^{\Sc})$.
These quantities can be evaluated more explicitly: with the help of \eqref{111} and Statement~\hyperref[stat1b]{1b} we find that
\begin{subequations}
\begin{align}
	p_n(n^{\Sc}) & =\tr\left(\left(\hat{P}(\hat{N}^S = N-n^{\Sc})\otimes \hat{I}^{\Sc}\right)\hat{\rho}_\gmc\right)\\
	& =  \tr^{S}\left(\hat{P}(\hat{N}^S = N-n^{\Sc}) \tr^{\Sc}(\hat{\rho}_\gmc)\right)\\
	& \approx \frac{1}{Z_\gc} \tr^S \left(\hat{P}(\hat{N}^S=N-n^{\Sc}) \tr^{\Sc}\left(e^{-\beta(\hat{H}_0 - \mu \hat{N})}  \right)\right)\\
	&= \frac{1}{Z_\gc^S} e^{\beta\mu (N-n^{\Sc})}\tr^S\left(\hat{P}(\hat{N}^S=N-n^{\Sc}) e^{-\beta \hat{H}_0^{S}} \right),\label{ex2last1}
\end{align}
\end{subequations}
where $\hat{H}_0^S$ is the part of $\hat{H}_0$ acting only on $S$ and $\mu=\mu_{01}$. Likewise,
\begin{subequations}
\begin{align}
	\hat{\rho}(n^{\Sc}) &= \frac{1}{p_n(n^{\Sc})} \tr^{\Sc}\left(\hat{P}(\hat{N}^{\Sc}=n^{\Sc})\hat{\rho}_\gmc\right)\\
	&=\frac{1}{p_n(n^{\Sc})} \hat{P}(\hat{N}^S=N-n^{\Sc}) \Bigl(\tr^{\Sc}\hat{\rho}_\gmc\Bigr) \hat{P}(\hat{N}^S=N-n^{\Sc})\\
	&\approx \frac{1}{Z_\gc^S}\frac{1}{p_n(n^{\Sc})} \hat{P}(\hat{N}^S=N-n^{\Sc}) \, e^{-\beta(\hat{H}_0^S-\mu\hat{N}^S)} \, \hat{P}(\hat{N}^S=N-n^{\Sc})\\
	&= \frac{1}{Z_\gc^S  p_n(n^{\Sc})} e^{\beta\mu(N-n^{\Sc})} \hat{P}(\hat{N}^S=N-n^{\Sc})\, e^{-\beta\hat{H}_0^S} \, \hat{P}(\hat{N}^S=N-n^{\Sc}).\label{ex2last2}
\end{align}
\end{subequations}
Therefore, $\hat{\rho}(n^{\Sc})$ can be interpreted as a canonical density matrix of the system $S$ with $N-n^{\Sc}$ particles and thus \eqref{eq: law ex} is a mixture of GAP measures with canonical density matrices corresponding to different particle numbers in $S$.

As a by-product, we observe that \eqref{pdef} reduces in this example to \eqref{ex2last1} with $q_1^S=n^S=N-n^{\Sc}$ because $\hat{Q}_1=\hat{N}$, and \eqref{rhoSdef} reduces to \eqref{ex2last2}. Thus, the calculation above amounts to an independent confirmation of Remark~\ref{rem:4b+} in this case that does not assume the validity of the general Gibbs principle.\hfill$\diamond$
\end{ex}

\section{Derivation of Statement \hyperref[stat5]{5}: Approach to Equilibrium}
\label{sec:approach}

\subsection{Mathematical Tools}

For deriving Statement \hyperref[stat5]{5}, we make use of the following three mathematical propositions. 

The first one provides a large part of the first statement in Statement~\hyperref[stat5]{5}: it states that a version of the eigenstate thermalization hypothesis (ETH) suitable for microscopic thermal equilibrium (MITE) implies that the reduced density matrix $\hat{\rho}^S_{\Psi_t}$ approaches $\tr^{\Sc}\hat{\rho}_\gmc$ for every initial state $\Psi_0\in\mathbb{S}(\Hilbert_\gmc)$. The MITE-ETH was formulated in Statement~\hyperref[stat5]{5} as \eqref{ETHa} and \eqref{ETHb} and is formulated in Proposition~\ref{prop: MITE} in a more precise version as \eqref{ineq: ETH1} and \eqref{ineq: ETH2}.

For the mathematical formulation, we recall the precise meaning of ``most'' \cite{vN29,GLMTZ_10b}: One says that a statement $s(x)$ is true \emph{for $(1-\varepsilon)$-most $x\in X$} relative to the normalized measure $\mu$ on $X$ if and only if
\be
\mu\bigl\{x\in X:s(x)\bigr\}\geq 1-\varepsilon\,.
\ee
One says that a statement $s(t)$ is true \emph{for $(1-\varepsilon)$-most $t\in[0,\infty)$} if and only if
\be
\liminf_{T\to \infty} \frac{1}{T} \Bigl| \bigl\{t\in[0,\infty): s(t) \bigr\} \Bigr| \geq 1-\varepsilon\,,
\ee
where $|M|$ denotes the length (Lebesgue measure) of the set $M\subseteq \RRR$.

\begin{prop}\label{prop: MITE}
Let $\varepsilon,\delta>0$, $\Hilbert = \Hilbert^S \otimes\Hilbert^B$, let $\Hilbert_R\subset\Hilbert$ be a subspace, $\hat{P}_R$ the projection to $\Hilbert_R$, and $\hat{H}$ a Hamiltonian on $\Hilbert$ such that $[\hat{H},\hat{P}_R]=0$. Moreover, suppose (MITE-ETH) that for any eigenvectors $\phi_1,\phi_2\in\mathbb{S}(\Hilbert_R)$ of $\hat{H}$ with eigenvalues $e_1\neq e_2$,
\begin{align}
    \Bigl\|\tr^B |\phi_1\rangle\langle\phi_1| - \tr^B\hat{\rho}_R\Bigr\|_2 < \varepsilon,\label{ineq: ETH1}\\
    \Bigl\|\tr^B |\phi_1\rangle\langle\phi_2| \Bigr\|_2 <\varepsilon,\label{ineq: ETH2}
\end{align}
where $\|\cdot\|_2$ denotes the Hilbert-Schmidt norm. Then every $\psi_0\in\mathbb{S}(\Hilbert_R)$ is such that for $(1-\delta)$-most $t\in[0,\infty)$,
\begin{align}
    \Bigl\|\tr^B|\psi_t\rangle\langle\psi_t|-\tr^B\hat{\rho}_R \Bigr\|_2 < 2\varepsilon\sqrt{\frac{D_G}{\delta}}\,,
\end{align}
where $D_G$ is the maximal degeneracy of eigenvalue gaps, 
\be\label{DGdef}
D_G := \max_{E\in\RRR} \# \bigl\{ (e,e')\in\mathcal{E}\times\mathcal{E}: e\neq e' \text{ and }e-e'=E \bigr\}\,,
\ee
and $\mathcal{E}$ the set of eigenvalues of $\hat{H}$.
\end{prop}

Proposition~\ref{prop: MITE} expresses that for every initial state $\psi_0\in\mathbb{S}(\Hilbert_R)$, for most times $t\geq 0$, $\hat{\rho}^S_{\psi_t} \approx \tr^B\hat{\rho}_R$. The proof can be found in Appendix~\ref{app:proof}. It is an adaption and generalization of the one given in \cite{Sre_96,RDO_08,GHLT_17} for the case that $\hat{\rho}_R=\hat{\rho}_{\mc}$ and the Hamiltonian has non-degenerate eigenvalues and eigenvalue gaps.

Another proposition concerns the commutation of quantifiers. It is helpful to introduce the symbol $\formost$ to denote ``for most.'' If $A(x,y)$ is any statement about objects $x$ and $y$, then ``$\forall x \forall y:A(x,y)$'' is equivalent to ``$\forall y \forall x:A(x,y)$.'' Likewise, ``$\formost x \formost y:A(x,y)$'' is equivalent to ``$\formost y \formost x:A(x,y)$'' \cite[Footnote 7]{GLMTZ_10b} if both $\formost x$ and $\formost y$ refer to probability spaces (i.e., normalized measures); the situation is more subtle when talking about most times $t\in [0,\infty)$ because the uniform measure on $[0,\infty)$ is not normalizable. In fact, ``$\formost x \formost t: A(x,t)$'' implies that ``$\formost t\formost x:A(x,t)$'' but in general not vice versa \cite[Footnote 5]{Vog_24}. The following proposition provides a sufficient condition for commuting $\formost t$ from the left to the right of $\formost x$. 

\begin{prop}\label{prop:vertauschen}
Suppose $(X,\mathscr{F},\PPP)$ is a probability space, $C>0$, $F:X\times [0,\infty) \to [0,C]$ is a measurable function, $0<\varepsilon<1$, $0<\delta_x<1$, and $0<\delta_t<1$. Suppose further that for every $x\in X$, the time average
\be\label{Tlimexists}
\overline{F(x,t)}:=\lim_{T\to\infty} \frac{1}{T} \int_0^T dt \, F(x,t)\text{ exists.}
\ee
If $(1-\delta_t)$-most $t\in[0,\infty)$ are such that for $(1-\delta_x)$-most $x\in X$, $F(x,t)\leq \varepsilon$, then $(1-\delta'_x)$-most $x\in X$ are such that for $(1-\delta'_t)$-most $t\in[0,\infty)$, $F(x,t)\leq \varepsilon'$ with any positive $\varepsilon', \delta_t', \delta'_x$ such that
\be\label{delta'x}
 \varepsilon' \delta'_t \delta'_x\leq\varepsilon + C (\delta_x + \delta_t)
\ee
(e.g., $\varepsilon'=\delta'_t=\delta'_x=(\varepsilon+C(\delta_x+\delta_t))^{1/3}$).
\end{prop}

The proof can be found in Appendix~\ref{app:vertauschen}.

\begin{prop}\label{prop:ergodic}
Let $\Hilbert$ be a finite-dimensional Hilbert space, $\Psi_0\in\SSS(\Hilbert)$, $\Psi_t=\exp(-i\hat{H}t)\Psi_0$ with self-adjoint $\hat{H}$, and let $g:\SSS(\Hilbert)\to \RRR$ be any bounded measurable function. Then the time average $\overline{g(\Psi_t)}$ exists.
\end{prop}

The proof can be found in Appendix~\ref{app:ergodic}.

\subsection{Application to Statement \hyperref[stat5]{5}}

{\it Derivation of Statement \hyperref[stat5]{5}.} We assume that $\hat{H}$  satisfies the ETH \eqref{ineq: ETH1}--\eqref{ineq: ETH2}, and that $D_G$ is not too large. Then Proposition~\ref{prop: MITE} for $\Hilbert_R = \Hilbert_{\gmc}$ yields that for every $\Psi_0\in\SSS(\Hilbert_\gmc)$, for most times $t\in [0,\infty)$, $\hat{\rho}_{\Psi_t}^S \approx \tr^{\Sc} \hat{\rho}_\gmc$. Statement~\hyperref[stat1a]{1a} tells us that $\tr^{\Sc}\hat{\rho}_\gmc \approx \hat{\rho}_\gG^S$. We thus obtain the first statement \eqref{5} of Statement~\hyperref[stat5]{5}.

We now turn to the conditional wave function. 
Theorem~2 in \cite{GLMTZ_15} shows that for every $t$, most orthonormal bases $b$ of $\Hilbert^{\Sc}$ are such that $\law_{\psi^S(t)}\approx \GAP_{\hat{\rho}^S_{\Psi_t}}$.
By the continuous dependence \eqref{GAPcont} of $\GAP_{\hat\rho}$ on $\hat\rho$, $\GAP_{\hat{\rho}^S_{\Psi_t}}\approx \GAP_{\hat{\rho}_\gG^S}$ for most $t$ and regardless of the basis $b$. Thus, 
\be
\forall \Psi_0\ \formost t\geq 0\ \formost b:\  \law_{\psi^S(t)}\approx \GAP_{\hat{\rho}_\gG^S}. 
\ee
We want to move the ``$\formost b$'' to the left of ``$\formost t$'' because we want to apply the same basis $b$ at all times $t$. As discussed before Proposition~\ref{prop:vertauschen}, even commuting $\formost t$ and $\formost b$ requires further assumptions, to which we turn now. We express that $\law_{\psi^S(t)}$ is close to $\GAP_{\hat{\rho}^S_{\gG}}$ by means of integrating them against arbitrary ``test'' functions $f:\SSS(\Hilbert^S)\to\RRR$ as in \eqref{muf}. Proposition~\ref{prop:vertauschen} for $X=ONB(\Hilbert^{\Sc})$, $x=b$, and $F(b,t)=|\law_{\psi^S(t)}(f) - \GAP_{\hat\rho^S_{\gG}}(f)|$ (which satisfies $0\leq F(b,t) \leq C=2\|f\|_\infty$) allows us to interchange $\formost t$ and $\formost b$, as its hypothesis \eqref{Tlimexists} is satisfied by virtue of Proposition~\ref{prop:ergodic} with $g(\Psi)=|\law_{\psi^S}(f)-\GAP_{\hat\rho^S_{\gG}}(f)|$ (note that $\psi^S$ depends on $\Psi$) for arbitrary but fixed $b$ and $f$; after all, $F(b,t)=g(\Psi_t)$. Thus, 
\be\label{mostbmosttlaw}
\forall \Psi_0\ \formost b \ \formost t\geq 0 : \ \law_{\psi^S(t)}\approx \GAP_{\hat{\rho}_\gG^S}\,, 
\ee
which is what we claimed in the case of typical ONBs $b$; the case of ONBs diagonalizing $\hat{Q}_1,\ldots,\hat{Q}_{K-1}$ is analogous, while $\formost b$ now means ``for most ONBs $b$ diagonalizing $\hat{Q}_1,\ldots,\hat{Q}_{K-1}$'' instead of ``for most ONBs $b$.'' This completes the derivation of Statement~\hyperref[stat5]{5}.

An alternative strategy of the derivation starts from the fact, mentioned above, that $\forall \Psi_0\, \forall t\geq 0\, \formost b: \: \law_{\psi^S(t)}\approx \GAP_{\hat\rho_{\Psi_t}^S}$. Now we try first, before replacing $\hat\rho_{\Psi_t}^S$ by $\hat\rho_{\gG}^S$, to move ``$\formost b$'' to the left of the $t$-quantifier. It must still be expected that ``$\forall t$'' has to be replaced by the weaker ``$\formost t$'' because it may well happen that every $b$ has some bad, exceptional $t$'s \cite{GLMTZ_10b}. We apply Proposition~\ref{prop:ergodic} to $g(\Psi)=|\law_{\psi^S}(f)-\GAP_{\hat\rho_{\Psi}^S}(f)|$, then Proposition~\ref{prop:vertauschen} to $X=ONB(\Hilbert^{\Sc})$, $x=b$, and $F(b,t)=|\law_{\psi^S(t)}(f)-\GAP_{\hat\rho_{\Psi_t}^S}(f)|=g(\Psi_t)$, so we can interchange $\formost t$ and $\formost b$. Now if each of two statements $s_1(t)$ and $s_2(t)$ is true for $(1-\varepsilon)$-most $t$, then ``$s_1(t)$ and $s_2(t)$'' is true for at least $(1-2\varepsilon)$-most $t$. Since $\formost t: \hat\rho_{\Psi_t}^S \approx \hat\rho_\gG^S$ by \eqref{5} and for those $t$, $\GAP_{\hat\rho_{\Psi_t}^S}\approx \GAP_{\hat\rho_\gG^S}$ by \eqref{GAPcont}, we obtain \eqref{mostbmosttlaw}, as claimed.

\appendix

\section{Proof of Remark~\ref{rmk:differentF}}
\label{app:differentF}

We first show that if we had chosen $F$ differently, it would not have changed $\hat\rho_\gG$. Indeed, any other $\widetilde{F}$ would have been of the form $\widetilde{F}=GF$ with $G:\mathbb{R}^{K-1}\to \mathbb{R}^{K-1}$ an invertible linear mapping, so that (with $\lambda=(\lambda_1,\ldots,\lambda_{K-1})$ and $\cdot$ the dot product in $\mathbb{R}^{K-1}$)
\begin{align}
\sum_{k=0}^{K-1}\lambda_k F_k(\hat{N}_1,\ldots,\hat{N}_r)&= \lambda \cdot F(\hat{N}_1,\ldots,\hat{N}_r)\\ 
&= ((G^{-1})^\dagger \lambda)\cdot (GF(\hat{N}_1,\ldots,\hat{N}_r))\\[3mm]
&= \tilde{\lambda}\cdot \widetilde{F}(\hat{N}_1,\ldots,\hat{N}_r)
\end{align}
with
\be
\tilde\lambda=(G^{-1})^\dagger \lambda \,.
\ee
Thus, $\widetilde{\hat\rho}_\gG$ obtained from $\widetilde{F}$ and $\tilde\lambda$ coincides with $\hat\rho_\gG$ obtained from $F$ and $\lambda$; it remains to check that the $\lambda$ values obtained from $\widetilde{F}$ agree with $\tilde\lambda$. Since $\widetilde{\hat{Q}}_k=\widetilde{F}_k(\hat{N}_1,\ldots,\hat{N}_r)= G_k (\hat{Q}_1,\ldots,\hat{Q}_{K-1})$ and $\widetilde{Q}_k = G_k(Q_1,\ldots, Q_{K-1})$, we obtain from \eqref{fixlambda} for $k=1,\ldots,K-1$ by applying the linear mapping $G$ on both sides that $\tr(\widetilde{\hat\rho}_\gG \widetilde{\hat{Q}}_k)=\widetilde{Q}_k$, which completes the proof.

Now we show that if we had chosen $F$ differently, it would have changed $\hat\rho_\gmc$. Indeed, the axiparallel rectangle $\prod_{k=1}^{K-1} [Q_k-\Delta Q_k, Q_k]$ would be mapped by $G$ to a parallelepiped instead of the axiparallel rectangle $\prod_{k=1}^{K-1} [\widetilde{Q}_k-\Delta \widetilde{Q}_k, \widetilde{Q}_k]$ with the same corner $(\widetilde{Q}_1,\ldots,\widetilde{Q}_{K-1})$. Thus, the exact set of eigenvalues would be different for the rectangle.\footnote{It is worth pointing out that if the density of states (jointly for $\hat{Q}_1,\ldots,\hat{Q}_K$) grows quickly with each variable, most joint eigenvalues in either set are expected to lie close to that corner, but still the opening angle at that corner would be different.}

\section{Proof of Remark~\ref{rmk:freedom}}
\label{app:freedom}

If $\tilde{E}_{01},\ldots,\tilde{E}_{0r}$ solve \eqref{nuE} as well, then $\E:=(\tilde{E}_{01}-E_{01},\ldots, \tilde{E}_{0r}-E_{0r})$ is orthogonal to $\mathcal{L}$; $\hat{H}$ gets changed to $\widetilde{\hat{H}}=\hat{H}+\sum_i \E_i \hat{N}_i$, analogously $\hat{H}_0$, and $E$ to $\tilde{E}=E+\sum_i \E_i n_{0i}$. (For the purpose of this consideration, we may imagine a second set of molecules $\widetilde{A}_i$ subject to reactions exactly analogous to \eqref{reaction2} but with different ground state energies $\tilde{E}_{0i}$; in practice, $\widetilde{\hat{H}}$ may be our guess at the correct Hamiltonian.) Since $(n_{\eq,1}-n_{01},\ldots,n_{\eq,r}-n_{0r})\in \mathcal{L}$ and $\E\perp \mathcal{L}$, also $\tilde{E}=E+\sum_i \E_i n_{\eq,i}$. At the same time, $F,\hat{Q}_k$, and $Q_k$ ($k\leq K-1$) do not change. It follows that the choices $\tilde\beta=\beta$ and $\tilde{\mu}_{0i}=\mu_{0i}-\E_i$ will lead to $\widetilde{\hat\rho}_{\gG}=\hat\rho_\gG=\hat\rho_\gc=\widetilde{\hat\rho}_\gc$ and thus satisfy \eqref{fixlambda} with $\tilde{Q}_k=Q_k$ ($k\leq K-1$) and $\tilde{Q}_K=\tilde{E}$ on the right-hand side and $\widetilde{\hat{Q}}_K=\widetilde{\hat{H}}$ included on the left. Furthermore, $\tilde{\mu}_{0i}$ and $\tilde{E}_{0i}$ satisfy \eqref{munucond}; since the correct values of $\tilde\beta$ and $\tilde{\mu}_{0i}$ are determined by \eqref{fixlambda} and \eqref{munucond}, this confirms that $\tilde\beta=\beta$ and $\tilde{\mu}_{0i}=\mu_{0i}-\E_i$, so $\hat\rho_\gc$ stays the same.
Moreover, $\hat\rho_\gmc$ stays the same under the condition stated in Remark~\ref{rmk:freedom} because $\Hilbert_\gmc$ stays the same (if we neglect $\hat{V}$) because, put briefly, any change of particle numbers within $\Hilbert_\gmc$ lies in $\mathcal{L}$ while $\E\perp \mathcal{L}$.

\section{Proof of Example~\ref{ex:spread}}
\label{app:spread}

\begin{lemma}
    Let $M\in\mathbb{N}$. The function $\varrho(q)=C \,1_{0<q<Q} \, q^M$ is a probability density for $C=(M+1)/Q^{M+1}$ and has standard deviation $\sigma=Q \sqrt{M+1}/(M+2)\sqrt{M+3}$ (asymptotically $Q/M$ as $M\to\infty$). The function $\tilde\varrho(q) = \tilde{C}_M 1_{0<q} \, q^M\,  e^{-\lambda q}$ is a probability density for $\lambda>0$ and $\tilde{C}_M=M!/\lambda^{M+1}$, has expectation $Q$ for $\lambda=(M+1)/Q$, and has standard deviation $\tilde{\sigma}=Q/ \sqrt{M+1}$ (asymptotically $Q/\sqrt{M}$ as $M\to\infty$).
\end{lemma}

\begin{proof}
This is a straightforward calculation that we include for the convenience of the reader.
The value of $C$ is determined by
\be
1=\int_0^Q dq \, C q^M = C\frac{Q^{M+1}}{M+1}\,.
\ee
The mean of $\varrho$ is
\be
\langle q\rangle_{\varrho} = \int_0^Q dq\, \frac{M+1}{Q^{M+1}} \, q^{M+1}
= \frac{(M+1)Q^{M+2}}{Q^{M+1}(M+2)}= \frac{M+1}{M+2}Q\,,
\ee
the second moment
\be
\langle q^2 \rangle_{\varrho} = \int_0^Q dq \, \frac{M+1}{Q^{M+1}} \, q^{M+2} 
= \frac{M+1}{M+3}Q^2\,. 
\ee
Thus, the variance is
\begin{subequations}
\begin{align}
\mathrm{Var}_\varrho 
&= \langle q^2 \rangle_\varrho - \langle q\rangle_\varrho^2\\
&= Q^2 \Bigl[\frac{M+1}{M+3} - \frac{(M+1)^2}{(M+2)^2} \Bigr]\\ 
&= Q^2 \frac{M+1}{(M+3)(M+2)^2}\,,
\end{align}
\end{subequations}
which yields the standard deviation $\sigma$ in the lemma.

We turn to $\tilde\varrho$:
\begin{subequations}
\begin{align}
\frac{1}{\tilde{C}_M} 
&= \int_0^\infty dq \,  q^M \exp(-\lambda q) \\
&= \Bigl[-q^M \exp(-\lambda q)/\lambda\Bigr]_0^\infty + \frac{M}{\lambda} \int_0^\infty dq \, q^{M-1} \exp(-\lambda q) \\
&= \frac{M}{\lambda} \int_0^\infty dq \, q^{M-1} \exp(-\lambda q) = \ldots = \\
&= \frac{M!}{\lambda^{M+1}} \,.
\end{align}
\end{subequations}
The mean of $\tilde\varrho$ is
\begin{subequations}
\begin{align}
\langle q \rangle_{\tilde\varrho} 
&= \int_0^\infty dq \, \tilde{C}_M \, q^{M+1} \exp(-\lambda q)\\ 
&= \frac{\tilde{C}_M}{\tilde{C}_{M+1}}= \frac{\lambda^{M+1}(M+1)!}{M!\lambda^{M+2}}=\frac{M+1}{\lambda}\,.
\end{align}
\end{subequations}
Setting it equal to $Q$ yields $\lambda=(M+1)/Q$. The second moment is
\begin{subequations}
\begin{align}
\langle q^2 \rangle_{\tilde\varrho} 
&= \int_0^\infty dq \, \tilde{C}_M \, q^{M+2} \exp(-\lambda q)\\
&=\frac{\tilde{C}_M}{\tilde{C}_{M+2}}= \frac{\lambda^{M+1}(M+2)!}{M!\lambda^{M+3}}=\frac{(M+1)(M+2)}{\lambda^2}
= \frac{M+2}{M+1}Q^2\,.
\end{align}
\end{subequations}
Thus, the variance is
\begin{subequations}
\begin{align}
\mathrm{Var}_{\tilde\varrho} &= \langle q^2 \rangle_{\tilde\varrho} - \langle q\rangle_{\tilde\varrho}^2\\
&= \frac{M+2}{M+1}Q^2 - Q^2= \frac{Q^2}{M+1}\,,
\end{align}
\end{subequations}
which yields the standard deviation $\tilde\sigma$ mentioned in the lemma.
\end{proof}

\section{Proof of Proposition~\ref{prop:Qb}}
\label{app:Qb}

Existence: By \eqref{Qb},
\be
\hat{I}^a \otimes \hat{Q}^{bc} - \hat{I}^a \otimes \hat{I}^b \otimes \hat{Q}^c = \hat{Q}^{ab} \otimes \hat{I}^c - \hat{Q}^a \otimes \hat{I}^b \otimes \hat{I}^c\,.
\ee
Since the left-hand side is of the form $\hat{I}^a\otimes \hat{R}^{bc}$ and the right-hand side of the form $\hat{R}^{ab}\otimes \hat{I}^c$, both sides have to be of the form $\hat{I}^a\otimes \hat{Q}^b \otimes \hat{I}^c$. For the left-hand side, it follows (after dropping a factor $\hat{I}^a$) that $\hat{Q}^{bc}=\hat{Q}^b\otimes \hat{I}^c+ \hat{I}^b \otimes \hat{Q}^c$. For the right-hand side, it follows (after dropping a factor $\hat{I}^c$) that $\hat{Q}^{ab}=\hat{Q}^a\otimes \hat{I}^b + \hat{I}^a \otimes \hat{Q}^b$, as claimed.

Uniqueness: Insert \eqref{Qbc} into \eqref{Qb}, so $\hat{I}^a \otimes \hat{I}^b \otimes \hat{Q}^c$ cancels out; drop the common factor $\hat{I}^c$ and take the partial trace over $\Hilbert^a$ to obtain that $\hat{Q}^b=(\dim \Hilbert^a)^{-1}(\tr^a \hat{Q}^{ab}-(\tr \hat{Q}^a)\hat{I}^b)$.

\section{Proof of Proposition~\ref{prop: MITE}\label{app:proof}}

Let $\hat{\Pi}_e$ denote the projection to the eigenspace of $\hat{H}$ with eigenvalue $e\in \mathcal{E}$. We first show that for all $\psi_0\in\mathbb{S}(\Hilbert_R)$, the time average (defined in general as in \eqref{Tlimexists}) of $\tr^B |\psi_t\rangle\langle\psi_t|$ is close to $\tr^B\hat{\rho}_R$ in Hilbert-Schmidt norm. Since $\psi_t = \sum_{e\in\mathcal{E}} e^{-iet}\hat{\Pi}_e \psi_0$, the time average of $\tr^B|\psi_t\rangle\langle\psi_t|$ is given by
\begin{subequations}
\begin{align}
    \overline{\tr^B|\psi_t\rangle\langle\psi_t|} &= \tr^B \sum_{e,e'\in\mathcal{E}} \overline{e^{i(e-e')t}} \hat{\Pi}_{e'} |\psi_0\rangle\langle\psi_0|\hat{\Pi}_{e}\\
    &=\sum_{e\in\mathcal{E}} \tr^B \hat{\Pi}_e|\psi_0\rangle\langle\psi_0|\hat{\Pi}_e.
\end{align}
\end{subequations}
Without loss of generality we assume that $\hat{\Pi}_e\psi_0\neq 0$ for all $e$.
Since $\hat{H}$ and $\hat{P}_R$ are simultaneously diagonalizable, $\hat{\Pi}_e\psi_0\in\Hilbert_R$; thus, since $\phi_1:=\hat{\Pi}_e \psi_0/\|\hat{\Pi}_e\psi_0 \|$ is an eigenvector, \eqref{ineq: ETH1} implies that
\begin{subequations}
\begin{align}
    \left\|\overline{\tr^B |\psi_t\rangle\langle\psi_t|}- \tr^B\hat{\rho}_R \right\|_2 &= \left\|\sum_{e\in\mathcal{E}} \|\hat{\Pi}_e \psi_0\|^2 \left(\tr^B\frac{\hat{\Pi}_e|\psi_0\rangle\langle\psi_0|\hat{\Pi}_e}{\|\hat{\Pi}_e\psi_0\|^2} - \tr^B\hat{\rho}_R\right)\right\|_2\\
    &\leq \sum_{e\in\mathcal{E}} \|\hat{\Pi}_e\psi_0\|^2 \left\|\tr^B \frac{\hat{\Pi}_e|\psi_0\rangle\langle\psi_0|\hat{\Pi}_e}{\|\hat{\Pi}_e\psi_0\|^2} -\tr^B\hat{\rho}_R\right\|_2\\
    &< \sum_{e\in\mathcal{E}} \|\hat{\Pi}_e\psi_0\|^2 \varepsilon\\
    &=\varepsilon.
\end{align}
\end{subequations}
Next we show that the time variance of $\tr^B|\psi_t\rangle\langle\psi_t|$ is small (again in Hilbert-Schmidt norm). To this end we compute (with the notation $|\hat{A}|^2=\hat{A} \hat{A}^\dagger$)
\begin{subequations}
\begin{align}
    &\overline{\left\|\tr^B|\psi_t\rangle\langle\psi_t| - \overline{\tr^B|\psi_t\rangle\langle\psi_t|} \right\|_2^2}\nonumber\\ 
    &\quad = \overline{\left\|\tr^B \sum_{e,e'} e^{i(e-e')t} \hat{\Pi}_{e'}|\psi_0\rangle\langle\psi_0|\hat{\Pi}_e -\overline{\tr^B|\psi_t\rangle\langle\psi_t|}\right\|_2^2}\\
    &\quad =\overline{\left\|\sum_{e\neq e'} e^{i(e-e')t} \tr^B\hat{\Pi}_{e'}|\psi_0\rangle\langle\psi_0|\hat{\Pi}_e \right\|_2^2}\\
    &\quad = \tr^S\overline{\left|\sum_{e\neq e'} e^{i(e-e')t} \tr^B\hat{\Pi}_{e'}|\psi_0\rangle\langle\psi_0|\hat{\Pi}_e\right|^2}\\
    &\quad = \tr^S \sum_{\substack{e\neq e'\\ e''\neq e'''}} \overline{e^{i(e-e'-e''+e''')t}} \Bigl(\tr^B \hat{\Pi}_{e'}|\psi_0\rangle\langle\psi_0|\hat{\Pi}_e\Bigr)\Bigl(\tr^B \hat{\Pi}_{e'''}|\psi_0\rangle\langle\psi_0|\hat{\Pi}_{e''}\Bigr)\\
    &\quad = \sum_{\substack{e\neq e'\\ e''\neq e'''}} \delta_{e-e'-e''+e''',0} \tr^S\left(\Bigl(\tr^B \hat{\Pi}_{e'}|\psi_0\rangle\langle\psi_0|\hat{\Pi}_e\Bigr)\Bigl(\tr^B \hat{\Pi}_{e'''}|\psi_0\rangle\langle\psi_0|\hat{\Pi}_{e''}\Bigr)\right).
\end{align}
\end{subequations}

With the help of the Cauchy-Schwarz inequality for $\tr^S$ we find that
\begin{subequations}
\begin{align}
    &\left|\tr^S\left(\Bigl(\tr^B \hat{\Pi}_{e'}|\psi_0\rangle\langle\psi_0|\hat{\Pi}_e\Bigr)\Bigl(\tr^B \hat{\Pi}_{e'''}|\psi_0\rangle\langle\psi_0|\hat{\Pi}_{e''}\Bigr)\right)\right|\nonumber\\
    &\qquad\leq \sqrt{\tr^S\left[\Bigl(\tr^B \hat{\Pi}_{e'}|\psi_0\rangle\langle\psi_0|\hat{\Pi}_e\Bigr)\Bigl(\tr^B \hat{\Pi}_{e'}|\psi_0\rangle\langle\psi_0|\hat{\Pi}_e\Bigr)^*\right]}\nonumber\\
    &\qquad\quad\cdot \sqrt{\tr^S\left[\Bigl(\tr^B \hat{\Pi}_{e'''}|\psi_0\rangle\langle\psi_0|\hat{\Pi}_{e''}\Bigr)^*\Bigl(\tr^B \hat{\Pi}_{e'''}|\psi_0\rangle\langle\psi_0|\hat{\Pi}_{e''}\Bigr)\right]}\\
    &\qquad= \Bigl\| \tr^B \hat{\Pi}_{e'}|\psi_0\rangle\langle\psi_0|\hat{\Pi}_e\Bigr\|_2 \,\Bigl\|\tr^B \hat{\Pi}_{e'''}|\psi_0\rangle\langle\psi_0|\hat{\Pi}_{e''}\Bigr\|_2.
\end{align}
\end{subequations}
For $\alpha=(e,e')\in\mathcal{G}:=\{(\tilde{e},\tilde{e}')\in\mathcal{E}\times\mathcal{E}, \tilde{e}\neq \tilde{e}'\}$ we define $v_\alpha = \|\tr^B\hat{\Pi}_{e'}|\psi_0\rangle\langle\psi_0|\hat{\Pi}_e\|_2$ and $\Delta_\alpha=e-e'$. Moreover, in a move we adopt from \cite{SF_12}, we define the matrix $R$ by
\begin{align}
    R_{\alpha\beta} := \delta_{\Delta_\alpha,\Delta_\beta}.
\end{align}
With this it follows that
\begin{subequations}
\begin{align}
    \overline{\left\|\tr^B|\psi_t\rangle\langle\psi_t| - \overline{\tr^B|\psi_t\rangle\langle\psi_t|} \right\|_2^2} &\leq  \sum_{\alpha,\beta} v_\alpha^* R_{\alpha\beta}
    v_\beta\\
    &\leq \|R\| \sum_{\alpha} |v_\alpha|^2\\
    &\leq \|R\| \sum_{e\neq e'} \Bigl\| \tr^B \hat{\Pi}_{e'}|\psi_0\rangle\langle\psi_0|\hat{\Pi}_e\Bigr\|^2_2.
\end{align}
\end{subequations}
We have that $\|R\| \leq \sqrt{\|R\|_1 \|R\|_\infty}$,  where $\|R\|_1:=\max_{\|u\|_1=1} \|Ru\|_1$ with $\|u\|_1=\sum_\alpha |u_\alpha|$ and $\|R\|_\infty:= \max_{\|u\|_\infty=1}\|Ru\|_\infty$ with $\|u\|_\infty=\max_\alpha |u_\alpha|$. As $R$ is Hermitian, $\|R\|_1=\|R\|_\infty$, and we have the bound
\begin{align}
    \|R\| \leq \max_\beta \sum_{\alpha} |R_{\alpha\beta}| = D_G.
\end{align}

With this and the help of the MITE-ETH we obtain
\begin{subequations}
\begin{align}
\overline{\left\|\tr^B|\psi_t\rangle\langle\psi_t| - \overline{\tr^B|\psi_t\rangle\langle\psi_t|} \right\|_2^2} &\leq D_G \sum_{e\neq e'} \|\hat{\Pi}_{e}\psi_0\|^2 \|\hat{\Pi}_{e'}\psi_0\|^2 \left\|\tr^B \frac{\hat{\Pi}_{e'}|\psi_0\rangle\langle\psi_0|\hat{\Pi}_e}{\|\hat{\Pi}_{e'}\psi_0\|\,\|\hat{\Pi}_{e}\psi_0\|} \right\|_2^2\\
&\leq D_G \sum_{e,e'} \|\hat{\Pi}_e\psi_0\|^2 \|\hat{\Pi}_{e'}\psi_0\|^2 \varepsilon^2\\
&= D_G\varepsilon^2.
\end{align}
\end{subequations}
Markov's inequality says that
\be\label{Markov}
\PPP(Y\geq a)\leq \EEE Y/a
\ee
for any random variable $Y\geq 0$ and $a>0$. It
implies that (with the notation $|M|$ for the length (Lebesgue measure) of the set $M\subseteq \RRR$)
\begin{subequations}
\begin{align}
    \liminf_{T\to\infty} &\frac{1}{T}\left|\left\{t\in[0,T] : \left\|\tr^B|\psi_t\rangle\langle\psi_t| - \overline{\tr^B|\psi_t\rangle\langle\psi_t|} \right\|_2 > \varepsilon\sqrt{\frac{D_G}{\delta}}\right\}\right|\nonumber\\
    &=  \liminf_{T\to\infty} \frac{1}{T}\left|\left\{t\in[0,T] : \left\|\tr^B|\psi_t\rangle\langle\psi_t| - \overline{\tr^B|\psi_t\rangle\langle\psi_t|} \right\|_2^2 > \varepsilon^2 \frac{D_G}{\delta}\right\}\right|\\
    &\leq \frac{\delta}{\varepsilon^2 D_G} \overline{\left\|\tr^B|\psi_t\rangle\langle\psi_t| - \overline{\tr^B|\psi_t\rangle\langle\psi_t|} \right\|_2^2}\\
    &\leq \delta.
\end{align}
\end{subequations}
Therefore and as a consequence of the triangle inequality, every $\psi_0\in\mathbb{S}(\Hilbert_R)$ is such that for $(1-\delta)$-most $t\in[0,\infty)$,
\begin{subequations}
\begin{align}
    \Bigl\|\tr^B|\psi_t\rangle\langle\psi_t| - \tr^B\hat{\rho}_R \Bigr\|_2&\leq \Bigl\|\tr^B|\psi_t\rangle\langle\psi_t| - \overline{\tr^B|\psi_t\rangle\langle\psi_t|}\Bigr\|_2 + \Bigl\|\overline{\tr^B|\psi_t\rangle\langle\psi_t|} - \tr^B\hat{\rho}_R \Bigr\|_2\\
    &\leq 2\varepsilon\sqrt{\frac{D_G}{\delta}},
\end{align}
\end{subequations}
which finishes the proof.

\section{Proof of Proposition~\ref{prop:vertauschen}}
\label{app:vertauschen}

If for any given $t$, $(1-\delta_x)$-most $x\in X$ are such that $F(x,t)\leq \varepsilon$, then, since for the exceptional $x$'s at least $F(x,t) \leq C$, we have for the $x$-average that
\be
\int_X \PPP(dx) \, F(x,t) \leq \varepsilon + C \delta_x\,.
\ee
By assumption, this bound holds for $(1-\delta_t)$-most $t\geq 0$.
Since for every $t$, $\int_X \PPP(dx) \, F(x,t) \leq C$, 
\be
\limsup_{T\to\infty} \frac{1}{T} \int_0^T dt \int_X \PPP(dx) \, F(x,t) \leq \varepsilon + C \delta_x + C\delta_t \,.
\ee
By Tonelli's or Fubini's theorem, we can change the order of integration:
\be
\limsup_{T\to\infty}  \int_X \PPP(dx) \, \frac{1}{T}\int_0^T dt \, F(x,t) \leq \varepsilon + C \delta_x + C\delta_t \,.
\ee
We now use assumption \eqref{Tlimexists}; since the limit has to lie in $[0,C]$, we can apply the dominated convergence theorem (with dominating function of $x$ given by the constant $C$), so we can interchange the limit and the $x$-integral and obtain that
\be
\int_X \PPP(dx) \, \lim_{T\to\infty} \frac{1}{T}\int_0^T dt \, F(x,t)
=\lim_{T\to\infty} \int_X \PPP(dx) \, \frac{1}{T}\int_0^T dt \, F(x,t) \leq \varepsilon + C \delta_x + C\delta_t\,.
\ee
We apply Markov's inequality \eqref{Markov} to $x$: for any $\varepsilon''>0$ for $(1-\delta'_x)$-most $x$,
\be\label{liminfx}
\lim_{T\to\infty} \frac{1}{T}\int_0^T dt \, F(x,t) < \varepsilon''
\ee
with
\be
 \delta'_x=\frac{\varepsilon + C \delta_x + C\delta_t}{\varepsilon''} \,.
\ee
Fix any $x$ for which \eqref{liminfx} holds. By Markov's inequality applied to $t$, for any $\varepsilon'>0$ for $(1-\delta'_t)$-most $t\geq 0$,
\be
F(x,t) \leq \varepsilon'
\ee
with
\be
\delta'_t = \frac{\varepsilon''}{\varepsilon'} \,.
\ee
Regarding $\varepsilon'$ as given and $\varepsilon''$ as defined by the last equation, we obtain the conclusion of Proposition~\ref{prop:vertauschen}.

\section{Proof of Proposition~\ref{prop:ergodic}}
\label{app:ergodic}

Suppose first that $\hat{H}$ has rationally independent eigenvalues (which occurs generically in the sense that it occurs with probability 1 if $\hat{H}$ is random with continuous distribution). Let $(\phi_n)$ be an eigen-ONB of $\hat{H}=\sum_n e_n |\phi_n\rangle \langle \phi_n|$ and expand $\Psi=\sum_n c_n \phi_n$ in this basis, so $\Psi_t = \sum_n e^{-ie_n t}c_n \phi_n$. Then it is known \cite{Wey_16,vN29} that the dynamics of $\Psi_t$ is ergodic on the torus $\mathscr{T}:=\{\Phi\in\Hilbert: |\langle \phi_n|\Phi\rangle|=|c_n|\}$ defined by the moduli of the energy coefficients. Therefore, for any bounded measurable $g$ and almost every $\Psi_0$, the time average of $g(\Psi_t)$ is equal to the phase average. Since the torus is compact and $g$ bounded, the phase average exists, so also the time average exists. In fact, since the dynamics is quasi-periodic, the time average exists for \emph{every} (rather than \emph{almost} every) $\Psi_0$.

But even if the eigenvalues of $\hat{H}$ are not rationally independent, so certain phase relations become periodic rather than quasi-periodic and the dynamics is not ergodic on $\mathscr{T}$, the dynamics is still ergodic on a lower-dimensional submanifold $\mathscr{S}$, and again the time average exists and equals the phase average over $\mathscr{S}$. This completes the proof of Proposition~\ref{prop:ergodic}.

An alternative strategy of proof, at least for almost every $\Psi_0$, is based on the theorem \cite{Rokh_49a,Rokh_49b} that a measure-preserving dynamical system with finite measure on a Lebesgue space can be decomposed into ergodic components. For ergodic systems for almost every initial condition, the time average of any bounded measurable function exists and is equal to the phase average. Thus, for almost every initial condition the time average exists. Since a Hilbert space is a complete metric space and thus a Lebesgue space, and since in finite dimension $u_{\SSS(\Hilbert)}$ is finite and preserved, the dynamical system of Proposition~\ref{prop:ergodic} is contained as a special case.

\bigskip

\noindent\textbf{Acknowledgments.} 
We thank an audience of the Rutgers University mathematical physics webinar for helpful discussion. C.V.~was supported by the Deutsche Forschungsgemeinschaft (DFG, German Research Foundation) -- TRR 352 -- Project-ID 470903074. Moreover, C.V.~acknowledges financial support by the ERC Starting Grant ``FermiMath" No.~101040991 and the ERC Consolidator Grant ``RAMBAS'' No. 10104424, funded by the European Union. Views and opinions expressed are however those of the authors only and do not necessarily reflect those of the European Union or the European Research Council Executive Agency. Neither the European Union nor the granting authority can be held responsible for them.

\bibliographystyle{plainurl}
\bibliography{Literature.bib}

\end{document}